\documentclass[11pt,a4paper]{article}

\usepackage{fullpage}
\usepackage{amsmath,amsthm,amsfonts,dsfont}
\usepackage{amssymb,latexsym,graphicx}
\usepackage{enumitem}
\usepackage{palatino}
\usepackage{mathpazo}
\usepackage{stmaryrd}
\usepackage{mathtools}
\usepackage{hyperref}
\usepackage[lined,boxed]{algorithm2e}
\usepackage{subfigure}

\renewcommand{\cal}[1]{\mathcal{#1}}
\renewcommand{\vec}[1]{\mathbf{#1}}
\newcommand{\pr}[2]{\left\langle{#1, #2}\right\rangle}
\newcommand{\lat}{\mathcal{L}}
\newcommand{\set}[1]{\left\{#1\right\}}
\newcommand{\ceil}[1]{\left\lceil #1 \right\rceil}

\newcommand{\round}[1]{\left\lceil #1 \right\rfloor}

\newcommand{\N}{\ensuremath{\mathbb{N}}}

\newcommand{\R}{\ensuremath{\mathbb{R}}}
\newcommand{\Z}{\ensuremath{\mathbb{Z}}}
\newcommand{\Q}{\ensuremath{\mathbb{Q}}}
\newcommand{\eqdef}{\mathbin{\stackrel{\rm def}{=}}}

\def\V{\cal{V}}
\def\vol{{\rm vol}}
\def\conv{{\rm conv}}
\def\VR{{\rm VR}}
\def\unif{{\rm Uniform}}
\def\laplace{{\rm Laplace}}
\def\eps{\varepsilon}

\def\d{{\rm d}}

\makeatletter
\def\imod#1{\allowbreak\mkern8mu({\operator@font mod}\,\,#1)}
\makeatother

\DeclareMathOperator{\E}{\mathbb{E}}
\DeclareMathOperator{\VAR}{\mathbb{VAR}}

\DeclareMathOperator*{\argmin}{arg\,min}

\newcommand{\class}[1]{\ensuremath{\mathsf{#1}}}

\newcommand\otilde{\ensuremath{\widetilde{O}}}

\newcommand{\NP}{\class{NP}}

\newcommand{\problem}[1]{\ensuremath{\mathsf{#1}}}

\newcommand{\CVP}{\ensuremath{\problem{CVP}}}
\newcommand{\SVP}{\ensuremath{\problem{SVP}}}
\newcommand{\CVPP}{\ensuremath{\problem{CVPP}}}

\newcommand{\poly}{\ensuremath{{\rm poly}}}
\newcommand{\polylog}{\ensuremath{{\rm polylog}}}
\newcommand{\enc}[1]{\ensuremath{{\rm enc}\left(#1\right)}}

\theoremstyle{plain}
\newtheorem{theorem}{Theorem}
\newtheorem{lemma}[theorem]{Lemma}

\newtheorem{claim}[theorem]{Claim}
\newtheorem{definition}[theorem]{Definition}

\newif\ifnotes\notestrue

\ifnotes
\usepackage{color}
\definecolor{mygrey}{gray}{0.50}
\newcommand{\notename}[2]{{\textcolor{mygrey}{\footnotesize{\bf (#1:} {#2}{\bf ) }}}}

\else

\newcommand{\notename}[2]{{}}

\fi

\newcommand{\dnote}[1]{}
\newcommand{\nnote}[1]{}

\title{Short Paths on the Voronoi Graph and Closest Vector Problem with Preprocessing}
\author{
  Daniel Dadush\thanks{Department of Computer Science, New York University, New York (USA). \texttt{dadush@cs.nyu.edu}}
  \and
  Nicolas Bonifas\thanks{LIX, \'Ecole Polytechnique, Palaiseau and IBM, Gentilly (France). \texttt{nicolas.bonifas@polytechnique.edu}}
}
\date{\today }

\begin{document}
\maketitle

\begin{abstract}
\small\baselineskip=9pt
Improving on the Voronoi cell based techniques
of~\cite{journal/siamjdm/SommerFS09,journal/siamjc/MV13}, we give a \emph{Las Vegas} $\otilde(2^n)$
expected time and space algorithm for $\CVPP$ (the preprocessing version of the Closest Vector
Problem, $\CVP$). This improves on the $\otilde(4^n)$ deterministic runtime of the Micciancio
Voulgaris algorithm~\cite{journal/siamjc/MV13} (henceforth MV) for $\CVPP$~\footnote{The MV
algorithm also solves $\CVP$, as the preprocessing can be computed in the same time bound.} at the
cost of a polynomial amount of randomness (which only affects runtime, not correctness). 

As in MV, our algorithm proceeds by computing a short path on the Voronoi graph of the lattice,
where lattice points are adjacent if their Voronoi cells share a common facet, from the origin to a closest lattice
vector. Our main technical contribution is a randomized procedure that given the Voronoi relevant
vectors of a lattice -- the lattice vectors inducing facets of the Voronoi cell -- as preprocessing
and any ``close enough'' lattice point to the target, computes a path to a closest lattice vector
of expected polynomial size. This improves on the $\otilde(2^n)$ path length given by the MV
algorithm. Furthermore, as in MV, each edge of the path can be computed using a single iteration
over the Voronoi relevant vectors. 

As a byproduct of our work, we also give an optimal relationship between geometric and path distance
on the Voronoi graph, which we believe to be of independent interest. 

%
%
\end{abstract}


\thispagestyle{empty}

\newpage

\setcounter{page}{1}

\section{Introduction}

An $n$ dimensional lattice $\lat$ in $\R^n$ is defined as all integer combinations of some basis
$B=(\vec{b}_1,\dots,\vec{b}_n)$ of $\R^n$. The most fundamental computational problems on lattices
are the Shortest and Closest Vector Problems, which we denote by $\SVP$ and $\CVP$ respectively.
Given a basis $B \in \R^{n \times n}$ of $\lat$, the $\SVP$ is to compute $\vec{y} \in \lat
\setminus \set{\vec{0}}$ minimizing $\|\vec{y}\|_2$, and the $\CVP$ is, given an additional target
$\vec{t} \in \R^n$, to compute a vector $\vec{y} \in \lat$ minimizing $\|\vec{t}-\vec{y}\|_2$
~\footnote{The $\SVP$ and $\CVP$ can be defined over any norm, though we restrict our attention here
to the Euclidean norm.}. 

The study of the algorithms and complexity of lattice problems has yielded many
fundamental results in Computer Science and other fields over the last three
decades. Lattice techniques were introduced to factor polynomials with rational
coefficients \cite{lenstra82:_factor} and to show the polynomial solvability of
integer programs with a fixed number of integer variables
\cite{lenstra82:_factor, lenstra83:_integ_progr_with_fixed_number_of_variab}. It
has been used as a cryptanalytic tool for breaking the security of knapsack
crypto schemes \cite{DBLP:journals/jacm/LagariasO85}, and in coding theory for
developing structured codes~\cite{DBLP:journals/jsac/Buda89} and asymptotically
optimal codes for power-constrained additive white Gaussian noise (AWGN)
channels~\cite{DBLP:journals/tit/ErezZ04}. Most recently, the security of
powerful cryptographic primitives such as fully homomorphic encryption
\cite{DBLP:conf/stoc/Gentry09, DBLP:conf/crypto/Gentry10, conf/itcs/BV14} have
been based on the worst case hardness of lattice problems.

\paragraph{{\bf The Closest Vector Problem with Preprocessing.}} In $\CVP$ applications, a common
setup is the need to solve many $\CVP$ queries over the same lattice but with varying targets. This
is the case in the context of coding over a Gaussian noise channel, a fundamental channel model in
wireless communication theory. Lattice codes, where the codewords correspond to a
subset of lattice points, asymptotically achieve the AWGN channel capacity (for fixed transmission
power), and maximum likelihood decoding for a noisy codeword corresponds (almost) exactly to a
$\CVP$ query on the coding lattice. In the context of lattice based public key encryption, in most
cases the decryption routine can be interpreted as solving an approximate (decisional) $\CVP$ over a
public lattice, where the encrypted bit is $0$ if the point is close and $1$ if it is far.

$\CVP$ algorithms in this setting (and in general), often naturally break into a preprocessing
phase, where useful information about the lattice is computed (i.e. short lattice vectors, a short
basis, important sublattices, etc.), and a query~/~search phase, where the computed advice is used
to answer $\CVP$ queries quickly. Since the advice computed during preprocessing is used across all
$\CVP$ queries, if the number of $\CVP$ queries is large the work done in the preprocessing phase
can be effectively ``amortized out''. This motivates the definition of the Closest Vector Problem
with Preprocessing ($\CVPP$), where we fix an $n$ dimensional lattice $\lat$ and measure only the
complexity of answering $\CVP$ queries on $\lat$ after the preprocessing phase has been completed
(crucially, the preprocessing is done before the $\CVP$ queries are known). To avoid trivial
solutions to this problem, i.e.~not allowing the preprocessing phase to compute a table containing
all $\CVP$ solutions, we restrict the amount of space (as a function of the encoding size of the
input lattice basis) needed to store the preprocessing advice.  

\paragraph{{\bf Complexity.}} While the ability to preprocess the lattice is very powerful, it was
shown in~\cite{Micciancio01} that $\CVPP$ is \NP-hard when the size of the preprocessing advice is
polynomial. Subsequently, approximation hardness for the gap version of $\CVPP$ (i.e.~approximately
deciding the distance of the target) was shown in
\cite{FeigeMicciancio04,Regev03B,DBLP:journals/cc/AlekhnovichKKV11}, culminating in a hardness
factor of $2^{\log^{1-\eps} n}$ for any $\eps>0$~\cite{DBLP:conf/stoc/KhotPV12} under the assumption
that $\NP$ is not in randomized quasi-polynomial time. On the positive side,
polynomial time algorithms for the approximate search version of $\CVPP$ were studied (implicitly)
in~\cite{DBLP:journals/combinatorica/Babai86, DBLP:journals/combinatorica/LagariasLS90}, where the
current best approximation factor $O(n/\sqrt{\log n})$ was recently achieved in~\cite{DadushSR14}.
For the gap decisional version of $\CVPP$, the results are better, where the current best
approximation factor is $O(\sqrt{n/\log n})$~\cite{DBLP:journals/jacm/AharonovR05}. 

\paragraph{{\bf Exact $\CVPP$ algorithms}.} Given the hardness results for polynomial sized
preprocessing, we do not expect efficient algorithms for solving exact $\CVPP$ for general lattices.
For applications in wireless coding however, one has control over the coding lattice, though
constructing coding lattices with good error correcting properties (i.e.~large minimum distance) for
which decoding is ``easy'' remains an outstanding open problem. In this context, the study of fast
algorithms for exact $\CVPP$ in general lattices can yield new tools in the context of lattice
design, as well as new insights for solving $\CVP$ without preprocessing.

The extant algorithms for exact $\CVPP$ are in fact also algorithms for $\CVP$, that is, the time to
compute the preprocessing is bounded by query / search time. There are currently two classes of
$\CVP$ algorithms which fit the preprocessing / search model (this excludes only the randomized sieving
approaches~\cite{DBLP:conf/stoc/AjtaiKS01,DBLP:conf/coco/AjtaiKS02}). 

The first class is based on lattice basis reduction~\cite{lenstra82:_factor}, which use a ``short''
lattice basis as preprocessing to solve lattice problems, that is polynomial sized preprocessing.
The fastest such algorithm is due to
Kannan~\cite{kannan87:_minkow_convex_body_theor_and_integ_progr}, with subsequent refinements
in~\cite{journal/tcs/Helfrich85, DBLP:journals/combinatorica/Babai86,
DBLP:conf/crypto/HanrotStehle07, conf/soda/MW14}, which computes a Hermite-Korkine-Zolatoreff basis (HKZ) during the
preprocessing phase in $\otilde(n^{\frac{n}{2e}})$\footnote{The $\otilde$ notation suppresses
polylogarithmic factors.} time and $\poly(n)$ space, and in the query phase uses a search tree to compute
the coefficients of the closest vector under the HKZ basis in $\otilde(n^{\frac{n}{2}})$ time and
$\poly(n)$ space. 

The second class, which are the most relevant to this work, use the Voronoi cell (see
Section~\ref{sec:prelims-vor} for precise definitions) of the lattice -- the centrally symmetric
polytope corresponding to the points closer to the origin than to other lattice points -- as
preprocessing, and were first introduced by Sommer, Feder and
Shalvi~\cite{journal/siamjdm/SommerFS09}. In~\cite{journal/siamjdm/SommerFS09}, they give an
iterative procedure that uses the facet inducing lattice vectors of the Voronoi cell (known as the
\emph{Voronoi relevant vectors}) to move closer and closer to the target, and show that this
procedure converges to a closest lattice vector in a finite number of steps. The number of Voronoi
relevant vectors is $2(2^n-1)$ in the worst-case (this holds for almost all lattices), and
hence Voronoi cell based algorithms often require exponential size preprocessing. Subsequently,
Micciancio and Voulgaris~\cite{journal/siamjc/MV13} (henceforth MV), showed how to compute the
Voronoi relevant vectors during preprocessing and how to implement the search phase such that each
phase uses $\otilde(4^n)$ time and $\otilde(2^n)$ space (yielding the first $2^{O(n)}$ time
algorithm for exact $\CVP$!).  

While Voronoi cell based $\CVPP$ algorithms require exponential time and space on general lattices,
it was recently shown in~\cite{arxiv/MGC14} that a variant of~\cite{journal/siamjdm/SommerFS09} can
be implemented in polynomial time for lattices of Voronoi's first kind -- lattices which admit a set
of $n+1$ generators whose Gram matrix is the Laplacian of a non-negatively weighted graph -- using
these generators as the preprocessing advice. Hence, it is sometimes possible to ``scale down'' the
complexity of exact solvers for interesting classes of lattices.   

\paragraph{{\bf Main Result.}} Our main result is a randomized $\otilde(2^n)$ expected time and
space algorithm for exact $\CVPP$, improving the $\otilde(4^n)$ (deterministic) running time of ${\rm
  MV}$. Our preprocessing is the same as ${\rm MV}$, that is we use the facet inducing lattice vectors of
the Voronoi cell, known as the \emph{Voronoi relevant} vectors (see Figure~\ref{fig:VRvectors}), as the
preprocessing advice, which in the worst case consists of $2(2^n-1)$ lattice vectors. Our main
contribution, is a new search algorithm that requires only an expected polynomial number of iterations
over the set of Voronoi relevant vectors to converge to a closest lattice vector, compared to
$\otilde(2^n)$ in ${\rm MV}$.

\begin{figure}
  \centering
  \def\svgwidth{0.6\textwidth}
  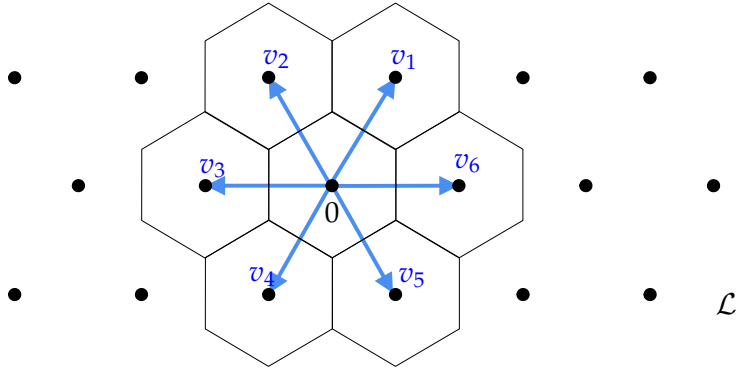
  \caption{Voronoi relevant vectors}
  \label{fig:VRvectors}
\end{figure}


One minor caveat to our iteration bound is that unlike that of ${\rm MV}$, which only depends on
$n$, ours also depends (at worst linearly) on the binary encoding length of the input lattice basis
and target (though the $\otilde(2^n)$ bound also holds for our procedure). Hence, while the bound is
polynomial, it is only ``weakly'' so. In applications however, it is rather anomalous to encounter $n$
dimensional lattice bases and targets whose individual coefficients require more than say $\poly(n)$
bits to represent, and hence the iteration bound will be $\poly(n)$ in almost all settings of
relevance. Furthermore, it is unclear if this dependence of our algorithm is inherent, or whether it
is just an artifact of the analysis.
 
While our algorithm is randomized, it is Las Vegas, and hence the randomness is in the runtime and
not the correctness. Furthermore, the amount of randomness we require is polynomial: it corresponds
to the randomness needed to generate a nearly-uniform sample from the Voronoi cell, which can be
achieved using Monte Carlo Markov Chain (MCMC) methods over convex bodies~\cite{DyerFK91,LV06}. This 
requires a polynomial number of calls to a membership oracle. Each membership oracle test requires an
enumeration over the $\otilde(2^n)$ Voronoi-relevant vectors, resulting in a total complexity of 
$\otilde(2^n)$.

Unfortunately, we do not know how to convert our $\CVPP$ improvement to one for $\CVP$. The
technical difficulty lies in the fact that computing the Voronoi relevant vectors, using the current
approach, is reduced to solving $\otilde(2^n)$ related lower dimensional $\CVP$s on an $n-1$
dimensional lattice (for which the Voronoi cell has already been computed). While the ${\rm MV}$
$\CVPP$ algorithm requires $\otilde(4^{n})$ for worst case targets (which we improve to
$\otilde(2^n)$), they are able to use the relations between the preprocessing $\CVP$s to solve each
of them in amortized $\otilde(2^n)$ time per instance. Hence, with the current approach, reducing
the running time of $\CVP$ to $\otilde(2^n)$ would require reducing the amortized per instance
complexity to polynomial, which seems very challenging.

\paragraph{{\bf Organization.}} In the next section,
section~\ref{sec:voronoi-navigation}, we explain how to solve $\CVPP$ by finding
short paths over the Voronoi graph. In particular, we review the iterative
slicer~\cite{journal/siamjdm/SommerFS09} and ${\rm
MV}$~\cite{journal/siamjc/MV13} algorithms for navigating the Voronoi graph, and
describe our new \emph{randomized straight line procedure} for this task. In
section~\ref{sec:path-length}, we state the guarantees for the randomized
straight line procedure and use it to give our expected $\otilde(2^n)$ time
$\CVPP$ algorithm (Theorem~\ref{thm:cvpp}), as well as an optimal relationship
between geometric and path distance on the Voronoi
graph~(Theorem~\ref{thm:geom-path-dist}). The main geometric estimates
underlying the analysis of the randomized straight path algorithm are proved in
section~\ref{sec:voronoi-proofs}. 

Definitions and references for the concepts and prior algorithms used in the
paper can be found in section~\ref{sec:prelims}. In particular, see
subsections~\ref{sec:lat-prelims} for basic lattice definitions, and
subsection~\ref{sec:prelims-vor} for precise definitions and fundamental facts
about the Voronoi cell and related concepts.


\section{Navigating the Voronoi graph}
\label{sec:voronoi-navigation}

In this section, we explain how one can solve $\CVP$ using an efficient navigation algorithm over
the Voronoi graph of a lattice. We first describe the techniques used
by~\cite{journal/siamjdm/SommerFS09,journal/siamjc/MV13} for finding short paths on this graph, and
then give our new (randomized) approach. 

\paragraph{{\bf Paths on the Voronoi graph.}} Following the strategy of
\cite{journal/siamjdm/SommerFS09,journal/siamjc/MV13}, our search algorithm
works on the Voronoi graph $\cal{G}$ of an $n$ dimensional lattice $\lat$. 

\begin{definition}[Voronoi Cell]
The Voronoi cell $\V(\lat)$ of $\lat$ is defined as
\begin{align*}
\V(\lat) &= 
\set{\vec{x} \in \R^n: \|\vec{x}\|_2 \leq \|\vec{x}-\vec{y}\|_2,~\forall \vec{y}
\in \lat \setminus \set{\vec{0}}} \\
&= \set{\vec{x} \in \R^n: \pr{\vec{x}}{\vec{y}} \leq 
\pr{\vec{y}}{\vec{y}}/2,~\forall \vec{y} \in \lat \setminus \set{\vec{0}}} \text{,}
\end{align*}
the set of points closer to $\vec{0}$ than any other lattice point. When
the lattice in question is clear, we simply write $\V$ for $\V(\lat)$. It was
shown by Voronoi that $\V$ is a centrally symmetric polytope with at most
$2(2^n-1)$ facets. We define $\VR$, the set of \emph{Voronoi relevant vectors}
of $\lat$, to be the lattice vectors inducing facets of $\V$. 

The Voronoi graph $\cal{G}$ is the contact graph induced by the tiling of space
by Voronoi cells, that is, two lattice vectors $\vec{x},\vec{y} \in \lat$ are
adjacent if their associated Voronoi cells $\vec{x}+\V$ and $\vec{y}+\V$ touch
in a shared facet (equivalently $\vec{x}-\vec{y} \in \VR$). We denote the
shortest path distance between $\vec{x},\vec{y} \in \lat$ on $\cal{G}$ by
$d_{\cal{G}}(\vec{x},\vec{y})$. 

See Section~\ref{sec:prelims-vor} for more basic facts about the Voronoi cell.
\end{definition}

To solve $\CVP$ on a target $\vec{t}$, the idea of Voronoi cell based methods is to compute a short
path on the Voronoi graph $\cal{G}$ from a ``close enough'' starting vertex $\vec{x} \in \lat$ to
$\vec{t}$ (usually, a rounded version of $\vec{t}$ under some basis), to the center $\vec{y} \in
\lat$ of a Voronoi cell containing $\vec{t}$, which we note is a closest lattice vector by
definition. (see Figure~\ref{fig:VoronoiCVP}).

\begin{figure}
  \centering
  \def\svgwidth{0.5\textwidth}
  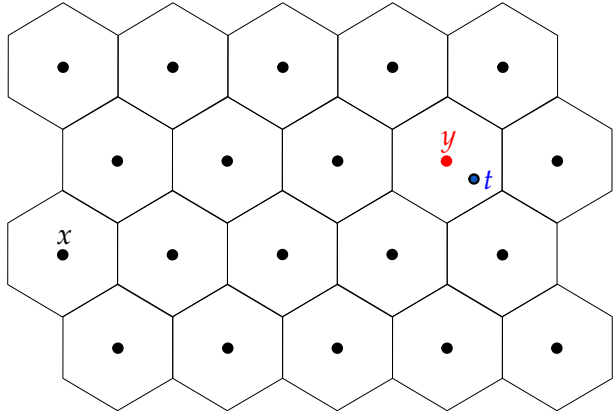
  \caption{CVP solution is the center of target-containing Voronoi cell}
  \label{fig:VoronoiCVP}
\end{figure}

\paragraph{{\bf Iterative slicer.}} The iterative slicer~\cite{journal/siamjdm/SommerFS09} was the first
$\CVP$ algorithm to make use of an explicit description of the Voronoi cell, in the form of the VR vectors.

The path steps of the iterative slicer are computed by greedily choosing any Voronoi relevant vector
that brings the current iterate $\vec{z} \in \lat$ closer to the target $\vec{t}$. That is, if there exists
a VR vector $\vec{v}$ such that $\|\vec{z}+\vec{v}-\vec{t}\|_2 < \|\vec{z}-\vec{t}\|_2$, then we
move to $\vec{z}+\vec{v}$. This procedure is iterated until there is no improving VR vector, at
which point we have reached a closest lattice vector to $\vec{t}$. This procedure was shown to
terminate in a finite number of steps, and currently, no good quantitative bound is known on its
convergence time.   

\paragraph{{\bf The Voronoi norm.}} We now make precise which notion of closeness to the target we
use (as well as MV) for the starting lattice vector $\vec{x}$ to the target $\vec{t}$. Notice that
for the path finding approach to make sense from the perspective of $\CVP$, we need to start the
process from a point $\vec{x} \in \lat$ that we know is apriori close in graph distance to
a closest lattice vector $\vec{y}$ to $\vec{t}$. Given the complexity of $\cal{G}$ and the fact that
we do not know $\vec{y}$, we will need a robust proxy for graph distance that we can estimate
knowing only $\vec{x}$ and $\vec{t}$. From this perspective, it was shown
in~\cite{journal/siamjc/MV13} that the Voronoi norm
\[
\|\vec{t}-\vec{x}\|_\V = \inf \set{s \geq 0: \vec{t}-\vec{x} \in s\V} = \sup_{\vec{v} \in \VR}
2\frac{\pr{\vec{v}}{\vec{t}-\vec{x}}}{\pr{\vec{v}}{\vec{v}}}
\]
of $\vec{t}-\vec{x}$ (i.e. the smallest scaling of $\V$ containing $\vec{t}-\vec{x}$) can be used to
bound the shortest path distance between $\vec{x}$ and $\vec{y}$. Here the quantity
$\|\vec{t}-\vec{x}\|_\V$ is robust in the sense that $\|\vec{y}-\vec{x}\|_\V \leq
\|\vec{t}-\vec{x}\|_\V + \|\vec{y}-\vec{t}\|_\V \leq \|\vec{t}-\vec{x}\|_\V + 1$ by the triangle
inequality. Hence from the perspective of the Voronoi norm, $\vec{t}$ is simply a ``noisy'' version
of $\vec{y}$. Furthermore, given that each Voronoi relevant vector has Voronoi norm $2$, one can
construct a lattice vector $\vec{x}$ such that $\|\vec{t}-\vec{x}\|_\V \leq n$, by simply expressing
$\vec{t} = \sum_{i=1}^n a_i \vec{v}_i$, for some linearly independent $\vec{v}_1,\dots,\vec{v}_n \in
\VR$, and letting $\vec{x} = \sum_{i=1}^n \round{a_i} \vec{v}_i$.

\paragraph{{\bf The MV Path.}} We now present the MV path finding approach, and give the
relationship they obtain between $\|\vec{t}-\vec{x}\|_\V$ and the path distance to a closest lattice
vector $\vec{y}$ to $\vec{t}$. 

The base principle of ${\rm MV}$~\cite{journal/siamjc/MV13} is similar to that of the iterative
slicer, but it uses a different strategy to select the next VR vector to follow, resulting in a
provably single exponential path length.

In ${\rm MV}$, a path step consists of tracing the straight line from the current path vertex
$\vec{z} \in \lat$ to the target $\vec{t}$, and moving to $\vec{z}+\vec{v}$ where $\vec{v} \in \VR$
induces a facet (generically unique) of $\vec{z}+\V$ crossed by the line segment
$[\vec{z},\vec{t}]$. It is not hard to check that each step can be computed using $O(n|\VR|) =
\otilde(2^n)$ arithmetic operations, and hence the complexity of computing the path is $O(n|\VR|
\times \text{ path length})$. 

The main bound they give on the path length, is that if the start vertex $\vec{x} \in 2\V+\vec{t}$
(i.e.~Voronoi distance less than $2$), then the path length is bounded by $2^n$. To prove the bound,
they show that the path always stays inside $\vec{t}+2\V$, that the $\ell_2$ distance to the target
monotonically decreases along the path (and hence it is acyclic), and that the number of lattice
vectors in the interior of $\vec{t}+2\V$ is at most $2^n$.

To build the full path, they run this procedure on the Voronoi graph for decreasing exponential
scalings of $\lat$~\footnote{The MV path is in fact built on a supergraph of the Voronoi graph,
which has edges corresponding to $2^i\VR$, $i \geq 0$.}, and build a path (on a supergraph of
$\cal{G}$) of length $O(2^n\log_2\|\vec{t}-\vec{x}\|_\V)$. One can also straightforwardly
adapt the MV procedure to stay on $\cal{G}$, by essentially breaking up the line segment
$[\vec{x},\vec{t}]$ in pieces of length at most $2$, yielding a path length of
$O(2^n\|\vec{x}-\vec{t}\|_\V)$. Since we can always achieve a starting distance of
$\|\vec{x}-\vec{t}\|_\V \leq n$ by straightforward basis rounding, note that the distance term is
lower order compared to the proportionality factor $2^n$. 

\paragraph{{\bf Randomized Straight Line.}} Given the $2^n$ proportionality factor between geometric
and path distance achieved by the MV algorithm, the main focus of our work will be to reduce the
proportionality factor to polynomial. In fact, will show the existence of paths of length
$(n/2)(\|\vec{t}-\vec{x}\|_\V+1)$, however the paths we are able to construct will be 
longer.

For our path finding procedure, the base idea is rather straightforward, we simply attempt to follow
the sequence of Voronoi cells on the straight line from the start point $\vec{x}$ to the target
$\vec{t}$. We dub this procedure the straight line algorithm. As we will show, the complexity of
computing this path follows the same pattern as ${\rm MV}$ (under certain genericity assumptions),
and hence the challenge is proving that the number of Voronoi cells the path crosses is polynomial.
Unfortunately, we do not know how to analyze this procedure directly. In particular, we are unable
to rule out the possibility that a ``short'' line segment (say of Voronoi length $O(1)$) may pass
through exponentially many Voronoi cells in the worst case (though we do not have any examples of
this). 


To get around the problem of having ``unexpectedly many'' crossings, we will make use of
randomization to perturb the starting point of the line segment. Specifically, we will use a
\emph{randomized straight line} path from $\vec{x} \in \lat$ to $\vec{t}$ which proceeds as follows
(see Figure~\ref{fig:RndStraightLine}):
\begin{enumerate}[label=(\Alph*)]
\item Move to $\vec{x}+Z$, where $Z \sim \unif(\V)$ is sampled uniformly from the Voronoi cell.
\item Follow the line from $\vec{x}+Z$ to $\vec{t}+Z$.   
\item Follow the line from $\vec{t}+Z$ to $\vec{t}$.
\end{enumerate}

\begin{figure}
  \centering
  \def\svgwidth{0.5\textwidth}
  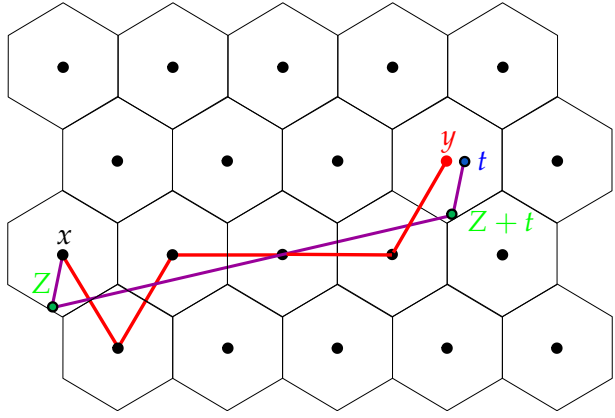
  \caption{Randomized Straight Line algorithm}
  \label{fig:RndStraightLine}
\end{figure}

We briefly outline the analysis bounding the expected number of Voronoi cells this path crosses,
which we claim achieves a polynomial proportionality factor with respect to $\|\vec{t}-\vec{x}\|_\V$. 

To begin, note that in phase A, we stay entirely within $\vec{x}+Z$, and hence do not cross any
Voronoi cells.

In phase B, at every time $\alpha \in [0,1]$, the point $(1-\alpha)\vec{x}+\alpha\vec{t}+Z$
is in a uniformly random coset of $\R^n / \lat$ since $Z$ is uniform. Hence the probability that we
we cross a boundary between time $\alpha$ and $\alpha+\eps$ is identical to the probability that we
cross a boundary going from $Z$ to $Z+\eps(\vec{t}-\vec{x})$. Taking the limit as $\eps \rightarrow
0$ and using linearity of expectation, we use the above invariance to show that the expected number
of boundaries we cross is bounded by $(n/2)\|\vec{t}-\vec{x}\|_\V$, the Voronoi distance between
$\vec{x}$ and $\vec{t}$. In essence, we relate the number of crossings to the probability that a
uniform sample from $\V$ (equivalently, a uniform coset) is close under the Voronoi norm to the
boundary $\partial\V$, which is a certain surface area to volume ratio.

Interestingly, as a consequence of our bound for phase B, we are able to give an optimal
relationship between the Voronoi distance between two lattice points and their shortest path
distance on $\cal{G}$, which we believe to be independent interest.  In particular, for two lattice
points $\vec{x},\vec{y} \in \lat$, we show in Theorem~\ref{thm:geom-path-dist} that the shortest
path distance on $\cal{G}$ is at least $\|\vec{x}-\vec{y}\|_\V/2$ and at most
$(n/2)\|\vec{x}-\vec{y}\|_\V$, which is tight for certain pairs of lattice points on $\Z^n$.

It remains now to bound the expected number of crossings in phase C. Here, the analysis is
more difficult than the second step, because the random shift is only on one side of the line
segment from $\vec{t}+Z$ to $\vec{t}$. We will still be able to relate the expected number of crossings to
``generalized'' surface area to volume ratios, however the probability distributions at each time
step will no longer be invariant modulo the lattice. In particular, the distributions become more
concentrated as we move closer to $\vec{t}$, and hence we slowly lose the benefits of the randomness
as we get closer to $\vec{t}$. Unfortunately, because of this phenomenon, we are unable to show in
general that the number of crossings from $\vec{t}+Z$ to $\vec{t}$ is polynomial. However, we will
be able to bound the number of crossings from $\vec{t}+Z$ to $\vec{t}+\alpha Z$ by
$O(n\ln(1/\alpha))$, that is, a very slow growing function of $\alpha$ as $\alpha \rightarrow 0$.
Fortunately, for rational lattices and targets, we can show that for $\alpha$ not too small, in
particular~$\ln(1/\alpha)$ linear in the size of binary encoding of the basis and target suffices,
that $\vec{t}+\alpha Z$ and $\vec{t}$ lie in the same Voronoi cell. This yields the claimed (weakly)
polynomial bound.


\section{Analysis and Applications of Randomized Straight Line}
\label{sec:path-length}

In this section, we give the formal guarantees for the randomized straight line algorithm and its
applications. The analysis here will rely on geometric estimates for the number of crossings, whose
proofs are found in Section~\ref{sec:voronoi-proofs}.  

To begin, we make formal the connection between Voronoi cells crossings, the length of the randomized
straight line path, and the complexity of computing it.

\begin{lemma}[Randomized Straight Line Complexity] Let $\vec{x} \in \lat$ be the starting point and let
$\vec{t} \in \R^n$ be the target. Then using perturbation $Z \sim \unif(\V)$, the expected edge length of
the path from $\vec{x}$ to a closest lattice vector $\vec{y}$ to $\vec{t}$ on $\cal{G}$ induced by the
randomized straight line procedure is
\[
\E[|(\lat + \partial\V) \cap [\vec{x}+Z,\vec{t}+Z]|] + 
\E[|(\lat + \partial\V) \cap [\vec{t}+Z,\vec{t})|] \text{ .}
\]
Furthermore, with probability $1$, each edge of the path can be computed using $O(n|\VR|)$
arithmetic operations.
\label{lem:rsl-complexity}
\end{lemma}

While rather intuitive, the proof of this Lemma is somewhat tedious, and so we
defer it to section~\ref{sec:appendix-path-length}. Note that $(\lat + \partial\V) \cap
[\vec{x}+Z,\vec{t}+Z]$ corresponds to the phase B crossings, and that $(\lat +
\partial\V) \cap [\vec{t}+Z,\vec{t})$ corresponds to the phase C crossings. 

Our bound for the phase B crossings, which is proved in Section~\ref{sec:phase-b}, is as
follows.

\begin{theorem}[Phase B crossing bound] Let $\lat$ be an $n$ dimensional lattice. Then for
$\vec{x},\vec{y} \in \R^n$ and $Z \sim {\rm uniform}(\V)$, we have that
\[
\E_Z[|(\lat + \partial\V) \cap [\vec{x}+Z,\vec{y}+Z]|] \leq (n/2)\|\vec{y}-\vec{x}\|_\V \text{ .}
\]
\label{thm:phase-b}
\vspace{-2em}
\end{theorem}

For phase C, we give a bound on the number crossings for a truncation of the phase C path. That is,
instead of going all the way from $\vec{t}+Z$ to $\vec{t}$, we stop at $\vec{t}+\alpha Z$, for
$\alpha \in (0,1]$. Its proof is given in Section~\ref{sec:phase-c}.

\begin{theorem}[Phase C crossing bound] For $\alpha \in (0,1]$, $Z \sim \unif(\V)$, $n \geq 2$, we have that
\[
\E[|(\lat + \partial\V) \cap [Z+\vec{t},\alpha Z+\vec{t}]|] 
\leq \frac{e^2}{\sqrt{2}-1} n (2 + \ln(4/\alpha)) \text{ .}
\]
\label{thm:phase-c}
\end{theorem}

 


Using the crossing estimate for phase B, we now show that from the perspective of existence, one can
improve the MV proportionality factor between geometric and path distance from exponential to linear
in dimension.  

\begin{theorem}\label{thm:geom-path-dist} For $\vec{x},\vec{y} \in \lat$, we have that
\[
(1/2) \|\vec{x}-\vec{y}\|_\V \leq d_{\cal{G}}(\vec{x},\vec{y}) \leq (n/2) \|\vec{x}-\vec{y}\|_\V \text{.}
\]
Furthermore, the above is best possible, even when restricted to $\lat = \Z^n$.
\end{theorem}
\begin{proof}
For the lower bound, note that $d_{\cal{G}}(\vec{x},\vec{y})$ is the minimum $k \in \Z_+$ such that
there exists $\vec{v}_1,\dots,\vec{v}_k \in \VR$ satisfying $\vec{y} = \vec{x} + \sum_{i=1}^k
\vec{v}_i$. Since $\forall \vec{v} \in \VR$, $\|\vec{v}\|_\V = 2$, by the triangle inequality
\[
\|\vec{y}-\vec{x}\|_\V = \|\sum_{i=1}^k \vec{v}_i\|_\V \leq \sum_{i=1}^k \|\vec{v}_i\|_\V = 2k \text{,}
\]
as needed.

For the upper bound, we run the randomized straight line procedure from $\vec{x}$ to $\vec{y}$,
i.e.~setting $\vec{t}=\vec{y}$. By Lemma~\ref{lem:rsl-complexity}, the expected path length on
$\cal{G}$ is
\[
\E[|(\lat+\partial\V) \cap [\vec{x}+Z,\vec{y}+Z]|] + \E[|(\lat+\partial\V) \cap [\vec{y}+Z,\vec{y})|] 
\]
where $Z \sim \unif(\V)$. Since $\vec{y} \in \lat$ and $Z \in {\rm int}(\V)$ with probability $1$,
note that
\[
\E[|(\lat+\partial\V) \cap [\vec{y}+Z,\vec{y})|] = 0\text{ ,}
\]
i.e.~the number of steps in phase C is $0$. It therefore suffices to bound the number phase B steps. By
Theorem~\ref{thm:phase-b}, we have that
\[
\E[|(\lat+\partial\V) \cap [\vec{x}+Z,\vec{y}+Z]|] \leq (n/2)\|\vec{x}-\vec{y}\|_\V\text{ ,}
\]
as needed. This shows the desired upper bound on the path length.

We now show that the above bounds are sharp. For the lower bound, note that it is tight for any two
adjacent lattice vectors, since $\forall \vec{v} \in \VR$, $\|\vec{v}\|_\V=2$. For the upper bound,
letting $\lat = \Z^n$, $\V = [-1/2,1/2]^n$, $\VR = \set{\pm \vec{e}_1,\dots, \pm \vec{e}_n}$, the
shortest path between $\vec{x} = \vec{0}$ and $\vec{y} = (1,\dots,1)$ has length $n$, while
$\|\vec{x}-\vec{y}\|_\V = 2 \|\vec{x}-\vec{y}\|_\infty = 2$.
\end{proof}

Since the Voronoi distance changes by at most $1$ when switching from $\vec{y}$ to $\vec{t} \in
\vec{y} + \V$, we note that the above bound immediately yields a corresponding bound on the path
length to a closest lattice vector to any target. 

As the phase C bound in Theorem~\ref{thm:phase-c} only holds for the truncated path, it does
yield a bound on the randomized straight line path length for general lattices. However, 
for rational lattices and targets, we now show that for $\alpha$ small enough, the truncated path
in fact suffices. 

We will derive this result from the following simple Lemmas.

\begin{lemma}[Rational Lattice Bound] Let $\lat \subseteq \Q^n$, and $\vec{t} \in \Q^n$. Let
$\bar{q} \in \N$ be the smallest number such that $\bar{q} \lat \subseteq \Z^n$ and $\bar{q} \vec{t}
\in \Z^n$, and let $\mu = \mu(\lat)$ denote the covering radius of $\lat$. For $\vec{y} \in \lat$,
if $\vec{t} \notin \vec{y}+\V$, then 
\[
\|\vec{t}-\vec{y}\|_\V \geq 1 + 1/(2\bar{q}\mu)^2 \text{ .}
\]
\label{lem:rat-lat-tar-bnd}
\end{lemma}
\begin{proof}
Note that $\vec{t} \notin \vec{y}+\V$ iff $\|\vec{t}-\vec{y}\|_\V > 1$. From here, we have that
\[
1 < \|\vec{t}-\vec{y}\|_\V = 2\frac{\pr{\vec{v}}{\vec{t}-\vec{y}}}{\pr{\vec{v}}{\vec{v}}} \text{,}
\]
for some $\vec{v} \in \VR$. By our assumptions, we note that $\pr{\vec{v}}{\vec{t}-\vec{y}} =
a/q^2$, for $a \in \N$. Next, $\|\vec{v}\|_2 \leq 2\mu$ (see the end of
section~\ref{sec:prelims-vor} for details) and $\vec{v} \in \Z^n/q$, and hence we can write
$\pr{\vec{v}}{\vec{v}} = b/q^2$, $b \in \N$, for $b \leq (2\bar{q}\mu)^2$. Therefore $1 <
\|\vec{t}-\vec{y}\|_\V = \frac{2a}{b}$ implies that $2a > b$. Since $a,b\in \N$, we must have that
\[
\|\vec{t}-\vec{y}\|_\V = \frac{2a}{b} \geq \frac{b+1}{b} = 1 + \frac{1}{b} \geq 1 + \frac{1}{(2\bar{q}\mu)^2}
\] 
as needed.
\end{proof}

The following shows that the relevant quantities in
Lemma~\ref{lem:rat-lat-tar-bnd} can be bounded by the binary encoding length of
the lattice basis and target. Since it is rather standard, we defer the proof to
section~\ref{sec:appendix-path-length}.

\begin{lemma}[Bit Length Bound]
Let $B \in \Q^{n \times n}$ be a lattice basis matrix for an $n$ dimensional lattice $\lat$, with
$B_{ij}
= \frac{p^B_{ij}}{q^B_{ij}}$ where $p^B_{ij} \in \Z$ and $q^B_{ij} \in \N$. Let $\vec{t} \in \Q^n$,
with
$\vec{t}_i = \frac{p^{\vec{t}}_i}{q^{\vec{t}}_i}$, $p^{\vec{t}}_i \in \Z$, $q^{\vec{t}}_i \in \N$.
Then
for $\bar{q} \in \N$, the smallest number such that $\bar{q} \lat \subseteq \Z^n$ and $\bar{q}
\vec{t}
\in \Z^n$, we have that $\log_2(\bar{q}\mu(\lat)) \leq \enc{B} + \enc{\vec{t}}$
and $\log_2(\mu(\lat)/\lambda_1(\lat)) \leq \enc{B}$.
\label{lem:bit-bnd}
\end{lemma}

%

We are now in a position to give our full $\CVPP$ algorithm. 

\begin{theorem}[$\CVPP$ Algorithm] Let $\lat$ be an $n$-dimensional lattice with basis $B \in \Q^{n
\times n}$, let $\VR$ denote the set of Voronoi relevant vectors of $\lat$. Given the set $\VR$ as
preprocessing, for any target $\vec{t} \in \Q^n$, a closest lattice vector to $\vec{t}$ can be
computed using an expected $\poly(n,\enc{B},\enc{t})|\VR|$ arithmetic operations.
\label{thm:cvpp}
\end{theorem}
\begin{proof}
To start we pick linearly independent $\vec{v}_1,\dots,\vec{v}_n \in \VR$.  We then compute the
coefficent representation of $\vec{t}$ with respect to $\vec{v}_1,\dots,\vec{v}_n$, that is $\vec{t}
= \sum_{i=1}^n a_i \vec{v}_i$. From here we compute the lattice vector $\vec{x} = \sum_{i=1}^n
\round{a_i} \vec{v}_i$, i.e.~the rounding of $\vec{t}$.

Next, using the convex body sampler (Theorem~\ref{thm:dfk-sampler}), we compute a
$(1/4)$-uniform sample $Z$ over $\V$. Note that a membership oracle for $\V$ can be
implemented using $O(n|\VR|)$ arithmetic operations. Furthermore, letting $\lambda_1
= \lambda_1(\lat)$, $\mu = \mu(\lat)$, we have that 
\[
(\lambda_1/2)B_2^n \subseteq \V \subseteq \mu B_2^n \text{ ,}
\]
where $\lambda_1 = \min_{\vec{v} \in \VR} \|\vec{v}\|_2$, $(1/2) \max_{\vec{v} \in \VR} \|\vec{v}\|_2
\leq \mu \leq (\sqrt{n}/2)\max_{\vec{v} \in \VR} \|\vec{v}\|_2$ (see Lemma~\ref{lem:vor-props} in
the Appendix). Hence, nearly tight sandwiching estimates for $\V$ can be easily computed using the
set $\VR$.

We now run the randomized straight line algorithm starting at lattice point $\vec{x}$, perturbation
$Z$, and target $\vec{t}$. If the path produced by the algorithm becomes longer than $c
n(n+(\enc{B}+\enc{\vec{t}}))$ (for some $c \geq 1$ large enough), restart the algorithm, and
otherwise return the found closest lattice vector.   

The correctness of the algorithm follows directly from the correctness of the randomized straight
line algorithm (Lemma~\ref{lem:rsl-complexity}), and hence we need only show a bound on the expected
runtime. 

\paragraph{{\bf Runtime.}} We first bound the number of operations performed in a single iteration.
Computing $\vec{v}_1,\dots,\vec{v}_n$, $\vec{t}$, and the sandwiching estimates for $\V$, requires
at most $O(n^3|\VR|)$ arithmetic operations. By Lemma~\ref{lem:bit-bnd}, the convex body sampler
requires at most
\[
\poly(n,\log(\sqrt{n}\mu/\lambda_1))|\VR| = \poly(n,\enc{B})|\VR|
\]
arithmetic operations. For the randomized straight line algorithm, each step requires at most
$O(n|\VR|)$ arithmetic operations by Lemma~\ref{lem:rsl-complexity}. Since we truncate it at
$O(n(n+(\enc{B}+\enc{\vec{t}})))$ iterations, this requires at most
$O(n^2(n+(\enc{B}+\enc{\vec{t}}))|\VR|)$ arithmetic operations. Hence the total number of arithmetic
operations per iteration is bounded by 
\[
\poly(n,\enc{B},\enc{t})|\VR| \text{. }
\]

We now show that the algorithm performs at most $O(1)$ iterations on expectation. For this it
suffices to show that each iteration succeeds with constant probability. In particular, we will show
that with constant probability, the length of the randomized straight line path is
bounded by $O(n^2(\enc{B}+\enc{\vec{t}}))$. To do this we will simply show that the expected path
length is bounded by $O(n^2(\enc{B}+\enc{\vec{t}}))$ under the assumption that $Z$ is truly uniform.
By Markov's inequality, the probability that the length is less than twice the expectation is at
least $1/2$ for a truly uniform $Z$, and hence it will be a least $1/4$ for a $1/4$-uniform $Z$.

To begin, we note that by the triangle inequality 
\[
\|\vec{t}-\vec{x}\|_\V \leq \sum_{i=1}^n |a_i-\round{a_i}|\|\vec{v}_i\|_\V \leq 
\sum_{i=1}^n (1/2)(2) = n \text{ .}
\]
Let $\bar{q}$ be as in Lemma~\ref{lem:bit-bnd}, and let $\alpha = \frac{1}{(4\bar{q}\mu)^2}$, where
we have that $\ln(1/\alpha) = O(\enc{B}+\enc{\vec{t}})$. Let $\vec{y} \in \lat$ denote the center of
the first Voronoi cell containing $\vec{t}+\alpha Z$ found by the randomized straight line
algorithm. We claim that $\vec{y}$ is a closest lattice vector to $\vec{t}$, or equivalently that $\vec{t}
\in \vec{y}+\V$. Assume not, then by Lemma~\ref{lem:rat-lat-tar-bnd}, $\|\vec{t}-\vec{y}\|_\V \geq 1
+ \frac{1}{(2\bar{q}\mu)^2}$. On the other hand, since $\vec{t}+\alpha Z \in \vec{y} + \V$ and $Z
\in \V$, by the triangle inequality
\[
\|\vec{t}-\vec{y}\|_\V \leq \|\vec{t}-(\vec{t}+\alpha Z)\|_\V + \|(\vec{t}+\alpha Z)-\vec{y}\|_\V
                       \leq \alpha + 1 = 1 + \frac{1}{(4\bar{q}\mu)^2} \text{ ,}
\] 
a contradiction. Hence $\vec{y}$ is a closest lattice vector to $\vec{t}$. If $Z \sim \unif(\V)$,
then by Theorems~\ref{thm:phase-b} and~\ref{thm:phase-c} the expected length of the randomized
straight line path up till $\vec{t}+\alpha Z$ (i.e.~till we find $\vec{y}$) is bounded by
\begin{align*}
(n/2)\|\vec{t}-\vec{x}\|_\V + \frac{e^2}{\sqrt{2}-1} n (2 + \ln(4/\alpha)) 
&= n^2/2 + \frac{e^2}{\sqrt{2}-1} n (2 + 2\ln(8\bar{q}\mu)) \\
&= O(n(n + (\enc{B}+\enc{t}))) \text{ ,}
\end{align*}
as needed. The theorem thus follows.
\end{proof}


\section{Preliminaries}
\label{sec:prelims}

\paragraph{{\bf Basics.}} 
For $n \geq 1$, we denote $\R^n,\Q^n,\Z^n$ to be the set of $n$ dimensional real / rational /
integral vectors respectively. We let $\N$ denote the set of natural numbers, and $\Z_+$ denote the
set of non-negative integers. For two sets $A,B \subseteq \R^n$, we denote their Minkowski sum $A+B
= \set{\vec{a} + \vec{b}: \vec{a} \in A, \vec{b} \in B}$. We write $\partial A$ to denote the
topological boundary of $A$. For a set $A \subseteq \R^n$, its affine hull, ${\rm affhull}(A)$, is the
inclusion wise smallest linear affine space containing $A$. We denote the interior of $A$ in $\R^n$
as ${\rm int}(A)$. We denote the relative interior of $A$ by ${\rm relint}(A)$, which
is the interior of $A$ with the respect to the subspace topology on ${\rm affhull}(A)$.  

For two $n$ dimensional vectors $\vec{x},\vec{y} \in \R^n$, we denote their inner product
$\pr{\vec{x}}{\vec{y}} = \sum_{i=1}^n \vec{x}_i\vec{y}_i$. The $\ell_2$ (Euclidean) norm of a vector
$\vec{x}$ is denoted $\|\vec{x}\|_2 = \sqrt{\pr{\vec{x}}{\vec{x}}}$. We let $B_2^n = \set{\vec{x}
\in \R^n: \|\vec{x}\|_2 \leq 1}$ denote the unit Euclidean ball, and let $S^{n-1} = \partial B_2^n$
denote the unit sphere. For vectors $\vec{x},\vec{y} \in \R^n$, we denote the closed line segment
from $\vec{x}$ to $\vec{y}$ by $[\vec{x},\vec{y}] \eqdef \set{\alpha \vec{x} + (1-\alpha)\vec{y}:
\alpha \in [0,1]}$, and $[\vec{x},\vec{y})$ the half open line segment not containing $\vec{y}$. 

We denote $\vec{e}_1,\dots,\vec{e}_n$ the vectors of the standard basis of $\R^n$, that is the vectors
such that $\vec{e}_i$ has a $1$ in the $i^\text{th}$ coordinate and $0$'s elsewhere.

\paragraph{{\bf Binary encoding.}} For an integer $z \in \Z$, the standard binary encoding for $z$
requires $1+\ceil{\log_2(|z|+1)}$ bits, which we denote $\enc{z}$.  For a rational number
$\frac{p}{q} \in \Q$, $p \in \Z$, $q \in \N$, the encoding size of $\frac{p}{q}$ is
$\enc{\frac{p}{q}} = \enc{p}+\enc{q}$. For an $n \times m$ matrix $M \in \Q^{m \times n}$ or vector
$\vec{a} \in \Q^n$, $\enc{M}$, $\enc{\vec{a}}$ denotes the sum of encoding lengths of all the
entries.

\paragraph{{\bf Integration.}} We denote the $k$-dimensional Lebesgue measure in $\R^n$ by
$\vol_k(\cdot)$. Only $k=n$ and $k=n-1$ will be used in this paper. For $k=n-1$, we will 
only apply it
to sets which can be written as a disjoint countable union of $n-1$ dimensional flat pieces. When
integrating a function $f:\R^n \rightarrow \R$ over a set $A \subseteq \R^n$ using the $n$
dimensional Lebesgue measure, we use the notation $\int_A f(\vec{x}) \d\vec{x}$. When integrating
with respect to the $n-1$ dimensional Lebesgue measure in $\R^n$, we write $\int_A f(\vec{x})
\d\vol_{n-1}(\vec{x})$.    

\paragraph{{\bf Probability.}} For a random variable $X \in \R$, we define its expectation by
$\E[X]$ and its variance by $\VAR[X] = \E[X^2]-\E[X]^2$. 
For two random variables $X,Y \in \Omega$, we define their total variation distance to be  
\[
d_{TV}(X,Y) = \max_{A \subseteq \Omega} |\Pr[X \in A]-\Pr[Y \in A]| \text{ .}
\] 

\begin{definition}[Uniform Distribution]
\label{def:unif}
For a set $A \subseteq \R^n$, we define the uniform distribution on $A$, denoted $\unif(A)$, to have
probability density function $1/\vol_n(A)$ and $0$ elsewhere. That is, for a uniform random variable
$X \sim \unif(A)$, we have that
\[
\Pr[X \in B] = \vol_n(A \cap B)/\vol_n(A)
\] 
for any measurable set $B \subseteq \R^n$.
\end{definition}

\paragraph{{\bf Complexity.}} We use the notation $\otilde(T(n))$ to mean $O(T(n)\polylog(T(n)))$.

\subsection{Lattices}
\label{sec:lat-prelims}

An $n$ dimensional lattice $\lat \subseteq \R^n$ is a discrete subgroup of $\R^n$ whose linear span
is $\R^n$. Equivalently, $\lat$ is generated by all integer combinations of some basis
$B=(\vec{b}_1,\dots,\vec{b}_n)$ of $\R^n$, i.e. $\lat = B \Z^n$. For $k \in \N$, we define the
quotient group $\lat / k\lat = \set{\vec{y}+k\lat: \vec{y} \in \lat}$. It is easy to check that the
map $\vec{a} \rightarrow B\vec{a} + k\lat$ from $(\Z / k \Z)^n \eqdef \Z_k^n$ to $\lat / (k\lat)$ is
an isomorphism. In particular $|\lat / (k\lat)| = k^n$. 

A shift $\lat+\vec{t}$ of $\lat$ is called a coset of $\lat$. The set of cosets of $\lat$ form a
group $\R^n / \lat$ under addition, i.e.~the torus. We will use the notation $A \imod{\lat}$, for a
set $A \subseteq \R^n$, to denote the set of cosets $\lat+A$. Note that $\R^n / \lat$ is isomorphic
to $[0,1)^n$ under addition $\imod{1}$ (coordinate wise), under the map $\vec{x} \rightarrow
B\vec{x} + \lat$ for any basis $B$ of $\lat$. We will need to make use of the uniform distribution
over $\R^n / \lat$, which we denote $\unif(\R^n / \lat)$. To obtain a sample from $\unif(\R^n /
\lat)$, one can take $U \sim \unif([0,1)^n)$ and return $BU \imod{\lat}$. 

We denote the length of the shortest non-zero vector (or minimum distance) of $\lat$ as $\lambda_1(\lat)
= \min_{\vec{y} \in \lat\setminus \set{\vec{0}}} \|\vec{y}\|_2$.  We denote the covering radius of $\lat$
as $\mu(\lat) = \max_{\vec{t} \in \R^n} \min_{\vec{y} \in \lat} \|\vec{t}-\vec{y}\|_2$ to be the farthest
distance between any point in space and the lattice.

The following standard lemma (see for instance~\cite{DBLP:journals/combinatorica/Babai86}) allows us
to bound the covering radius:
\begin{lemma}
\label{lem:rounding}
Let $\lat$ be an $n$-dimensional lattice. If $\vec{v}_1,\dots,\vec{v}_n \in \lat$ are linearly
independent lattice vectors, then $\mu(\lat) \leq \frac{1}{2} \sqrt{\sum_{i=1}^n \|\vec{v}_i\|^2}$.
\end{lemma}

%

\subsubsection{Voronoi cell, tiling, and relevant vectors}
\label{sec:prelims-vor}

For a point $\vec{t} \in \R^n$, let $\CVP(\lat,\vec{t}) = \argmin_{\vec{x} \in \lat}
\|\vec{t}-\vec{x}\|_2$, denote the set of closest lattice vectors to $\vec{t}$. For $\vec{y} \in \lat$, let
\[
H_\vec{y} = \set{\vec{x} \in \R^n: \|\vec{x}\|_2 \leq \|\vec{x}-\vec{y}\|_2} = \set{\vec{x} \in \R^n:
  \pr{\vec{y}}{\vec{x}} \leq \frac{1}{2}\pr{\vec{y}}{\vec{y}}} \text{ ,}
\]
denote the halfspace defining the set of points closer to $\vec{0}$ than to $\vec{y}$.

\begin{definition}[Voronoi Cell]
\label{def:vor-cell}
The Voronoi cell $\V(\lat)$ of $\lat$ is defined as
\[
\V(\lat) = \cap_{\vec{y} \in \lat \setminus \set{\vec{0}}} H_{\vec{y}} \text{ ,}
\]
the set of all points in $\R^n$ closer or at equal distance to the origin than to any other lattice 
point. 
\end{definition}

Naturally, $\V(\lat)$ is the set of points of $\lat$ whose closest lattice vector is $\vec{0}$.  We
abbreviate $\V(\lat)$ to $\V$ when the context is clear. It is easy to check from the definitions
that a vector $\vec{y} \in \lat$ is a closest lattice vector to a target $\vec{t} \in \R^n$ iff
$\vec{t}-\vec{y} \in \V$. The $\CVP$ is then equivalent to finding a lattice shift of $\V$
containing the target. 

From this, we see that the Voronoi cell tiles space with respect to $\R^n$, that is, the set of
shifts $\lat + \V$ cover $\R^n$, and shifts $\vec{x} + \V$ and $\vec{y} + \V$, $\vec{x},\vec{y} \in
\lat$, are interior disjoint if $\vec{x} \neq \vec{y}$. From the tiling property, we have the useful
property that the distribution $\unif(\V) \imod{\lat}$ is identical to $\unif(\R^n / \lat)$. 

We note that the problem of separating over the Voronoi cell reduces directly to $\CVP$, since if
$\vec{y} \in \lat$ is closer to a target $\vec{t}$ than $\vec{0}$, then $H_{\vec{y}}$
separates $\vec{t}$ from $\V$. Also, if no such closer lattice vector exists, then $\vec{t} \in \V$.

\begin{definition}[Voronoi Relevant Vectors]
\label{def:vr-vectors}
We define $\VR(\lat)$, the Voronoi relevant vectors of $\lat$, to be the minimal set of lattice vectors
satisfying $\V(\lat) = \cap_{\vec{v} \in \VR(\lat)} H_{\vec{v}}$, which we abbreviate to $\VR$ when the
context is clear. 
\end{definition}

Since the Voronoi cell is a full dimensional centrally symmetric polytope, the set
$\VR$ corresponds exactly to the set of lattice vectors inducing facets of $\V$ (i.e. such that $\V \cap
\partial H_{\vec{v}}$ is $n-1$ dimensional).

\begin{definition}[Voronoi Cell Facet]
For each $\vec{v} \in \VR$, let 
\[
F_\vec{v} = \V \cap \set{\vec{x} \in \R^n: \pr{\vec{x}}{\vec{v}} = \frac{1}{2}
\pr{\vec{v}}{\vec{v}}} \text{,}
\]
denote the facet of $\V$ induced by $\vec{v}$. 
\label{def:vr-facet}
\end{definition}

Here we have that 
\[
\partial \V = \bigcup_{\vec{v} \in \VR} F_{\vec{v}} \quad \text{ and } \quad 
\vol_{n-1}(\partial \V) = \sum_{\vec{v} \in \VR} \vol_{n-1}(F_{\vec{v}})
\]
since the intersection of distinct facets has affine dimension at most $n-2$. Similarly, 
\[
\V = \bigcup_{\vec{v} \in \VR} \conv(\vec{0},F_{\vec{v}}) \quad \text{ and } \quad 
\vol_n(\V) = \sum_{\vec{v} \in \VR} \vol_n(\conv(\vec{0},F_{\vec{v}})) \text{.}
\]

A central object in this paper will be $\lat+\partial\V$, the boundary of the lattice tiling.  We
shall call $\vec{y}+F_{\vec{v}}$, for $\vec{y} \in \lat,\vec{v} \in \VR$, a facet fo
$\lat+\partial\V$. Here, we see that
\[
\lat+\partial\V = \bigcup_{\vec{y} \in \lat, \vec{v} \in \VR} \vec{y}+F_{\vec{v}} \text{.}
\]
Note that each facet is counted twice in the above union, i.e.~$\vec{y}+F_{\vec{v}} =
(\vec{y}+\vec{v})+F_{-\vec{v}}$.

An important theorem of Voronoi classifies the set of Voronoi relevant vectors:

\begin{theorem}[Voronoi] For an $n$ dimensional lattice $\lat$, $\vec{y} \in \lat \setminus \set{\vec{0}}$ is in $\VR(\lat)$ if and only if
\[
\set{\pm \vec{y}} = \argmin_{\vec{x} \in 2\lat+\vec{y}} \|\vec{x}\|_2 \text{.}
\]
In particular, $|\VR| \leq 2(2^n-1)$.
\label{thm:vor-bound}
\end{theorem} 

Here the bound on $|\VR|$ follows from the fact that the map $\vec{y} \mapsto \vec{y} + 2\lat$ from $\VR$
to $\lat / (2\lat) \setminus \set{2\lat}$ is $2$-to-$1$. Furthermore, note that each Voronoi relevant
vector can be recovered from solutions to CVPs over $2\lat$. More precisely, given a basis $B$ for
$\lat$, each vector in $\vec{v} \in \VR$ can be expressed as $B\vec{p}-\vec{x}$, for some $\vec{p} \in
\set{0,1}^n \setminus \set{\vec{0}}$, and $\vec{x} \in \CVP(2\lat,B\vec{p})$ (we get a Voronoi relevant
iff $\vec{x}$ is unique up to reflection about $B\vec{p}$). 

We now list several important and standard properties we will need about the Voronoi cell and
relevant vectors. We give a proof for completeness. 

\begin{lemma} For an $n$ dimensional lattice $\lat$:
\begin{enumerate}
\item $\frac{\lambda_1(\lat)}{2} B_2^n \subseteq \V \subseteq \mu(\lat) B_2^n$.
\item $\lambda_1(\lat) = \min_{\vec{v} \in \VR} \|\vec{v}\|_2$.
\item $2\mu(\lat)/\sqrt{n} \leq \max_{\vec{v} \in \VR} \|\vec{v}\|_2 \leq 2\mu(\lat)$ 
\end{enumerate}
\label{lem:vor-props}
\end{lemma}
\begin{proof}
We prove each of the above in order:
\begin{enumerate}
\item Since each vector $\vec{y} \in \lat \setminus \set{\vec{0}}$ satisfies $\|\vec{y}\|_2 \geq
\lambda_1(\lat)$, we clearly have that $\lambda_1(\lat)/2B_2^n \subseteq H_{\vec{y}}$.  The
inner containment holds for $\V$ since $\V = \cap_{\vec{y} \in \lat \setminus \set{\vec{0}}} H_{\vec{y}}$.
For the outer containment, note that for any $\vec{t} \in \V$, that $\vec{0}$ is a closest lattice
vector to $\vec{t}$. Hence, by definition, $\|\vec{t}\|_2 = \|\vec{t}-\vec{0}\|_2 \leq \mu(\lat)$
as needed.
\item Since the set $\VR \subseteq \lat \setminus \set{\vec{0}}$, the vectors in $\VR$ clearly have
length greater than or equal to $\lambda_1(\lat)$. Next, let $\vec{y} \in \lat \setminus \set{\vec{0}}$ denote
a shortest non-zero vector of $\lat$. We wish to show that $\vec{y} \in \VR$. To do this, by
Theorem~\ref{thm:vor-bound}, we need only show that the only vectors of length $\lambda_1(\lat)$ in
$\vec{y}+2\lat$ are $\pm \vec{y}$. Assume not, then there exists $\vec{z} \in \vec{y} + 2\lat$, such
that $\vec{z}$ is not collinear with $\vec{y}$ having $\|\vec{z}\|_2 = \lambda_1(\lat)$.  But then
note that $(\vec{y}+\vec{z})/2 \in \lat \setminus\set{\vec{0}}$ and $\|(\vec{y}+\vec{z})/2\|_2
< \lambda_1(\lat)$, a contradition.
\item For $\vec{v} \in \VR$, we remember that $\vec{v} = \argmin_{\vec{z} \in 2\lat+\vec{v}}
\|\vec{z}\|_2$. In particular, this implies that $\|\vec{v}\|_2 \leq \mu(2\lat) = 2\mu(\lat)$ as
needed. Since the $\VR$ vectors span $\R^n$, we can find linearly independent
$\vec{v}_1,\dots,\vec{v}_n \in \VR$. By Lemma~\ref{lem:rounding}, we have that
\[
\mu(\lat) \leq (1/2)\sqrt{\sum_{i=1}^n \|\vec{v}\|_i^2} \leq \sqrt{n}/2 \max_{\vec{v} \in \VR}
\|\vec{v}\|_2 \text{ ,}
\]
as needed.
\end{enumerate}
\end{proof}

\subsection{Convex geometry}
\label{sec:prelims-convex}
A set $K \subseteq \R^n$ is a convex body, if it is convex (i.e.~$\vec{x},\vec{y} \in K \Rightarrow
[\vec{x},\vec{y}] \subseteq K$), compact and has non-empty interior. $K$ is symmetric if $K = -K$.
For a symmetric convex body $K \subseteq \R^n$, we define the norm (or gauge function) with respect
to $K$ by $\|\vec{x}\|_K = \inf \set{s \geq 0: \vec{x} \in sK}$, for any $\vec{x} \in \R^n$. A
function $f:K \rightarrow \R$ is convex (concave) if for all
$\vec{x},\vec{y} \in K$, $\alpha \in [0,1]$, 
\[
\alpha f(\vec{x}) + (1-\alpha)f(\vec{y}) \quad \geq (\leq) \quad 
f(\alpha \vec{x} + (1-\alpha)\vec{y}) \text{ .} 
\]
For a set $A \subseteq \R^n$, we define its convex hull $\conv(A)$ to be the
(inclusion wise) smallest convex set containing $A$. For two sets $A,B \subseteq \R^n$, we use the
notation $\conv(A,B) \eqdef \conv(A \cup B)$.   

For two non-empty measurable sets $A,B \subseteq \R^n$ such that $A+B$ is measurable, the
Brunn-Minkowski inequality gives the following fundamental lower bound
\begin{equation}
\label{thm:brunn-minkowski}
\vol_n(A+B)^{1/n} \geq \vol_n(A)^{1/n} + \vol_n(B)^{1/n} \text{.}
\end{equation}

\paragraph{{\bf Laplace Distributions.}} We define the Gamma function, $\Gamma(k) = \int_0^\infty
x^{k-1} e^{-x} dx$ for $k > 0$. For $k \in \N$, we note that $\Gamma(k) = (k-1)!$. We define the two
parameter distribution $\Gamma(k,\theta)$ on $\R$, $k,\theta \geq 0$, to have probability density
function $\frac{1}{\theta^k \Gamma(k)} x^{k-1} e^{-x/\theta}$, for $x \in \R$. For $r \sim
\Gamma(k,\theta)$, $k \in N$, the moments of $r$ are 
\[
\E[r^l] = \theta^l \frac{(k+l-1)!}{(k-1)!}\text{ , for } l \in \N \text{.} 
\]
In particular, $\E[r] = k\theta$ and $\VAR[r] = k\theta^2$. 

\begin{definition}[Laplace Distribution]
\label{def:laplace}
We define the probability distribution ${\rm Laplace}(K,\theta)$, with probability density function 
\[
f_K^\theta(\vec{x}) = \frac{\theta^n}{\vol_n(K)n!} e^{-\|x\|_K/\theta}, \quad \text{ for } \vec{x} \in \R^n.
\]
\end{definition}

Equivalently, a well known and useful fact (which we state without proof) is:
\begin{lemma} 
$X \sim {\rm Laplace}(K,\theta)$ is identically distributed to $rU$, 
where $r \sim \Gamma(n+1,\theta)$ and $U \sim \unif(K)$ are sampled independently.
\label{lem:lap-unif-gamma}
\end{lemma}

For our purposes, $\laplace(K,\theta)$ will serve as a ``smoothed'' out version of $\unif(K)$.  In
particular, letting $f$ denote the probability density function of $\laplace(K,\theta)$, for
$\vec{x},\vec{y} \in \R^n$, by the triangle inequality 
\begin{equation}
\label{lem:laplace-smooth}
\frac{f_K^\theta(\vec{x})}{f_K^\theta(\vec{y})} = \frac{e^{-\|\vec{x}\|_K/\theta}}{e^{-\|\vec{y}\|_K/\theta}}
\in [e^{-\|\vec{y}-\vec{x}\|_K/\theta},e^{\|\vec{y}-\vec{x}\|_K/\theta}] \text{.}
\end{equation}
Hence, the density varies smoothly as a function of $\|\cdot\|_K$ norm, avoiding the ``sharp''
boundaries of the uniform measure on $K$.

\paragraph{{\bf Algorithms.}}
A membership oracle $O_K$ for a convex body $K \subseteq \R^n$ is a function satisfying
$O_K(\vec{x}) = 1$ if $\vec{x} \in K$, and $O_K(\vec{x}) = 0$ otherwise. Most algorithms over convex
bodies can be implemented using only a membership oracle with some additional guarantees. 

In our $\CVPP$ algorithm, we will need to sample nearly uniformly from the Voronoi cell.  For this
purpose, we will utilize the classic geometric random walk method of Dyer, Frieze, and
Kannan~\cite{DyerFK91}, which allows for polynomial time near uniform sampling over any convex body.

\begin{theorem}[Convex Body Sampler~\cite{DyerFK91}]
\label{thm:dfk-sampler}
Let $K \subseteq \R^n$ be a convex body, given my a membership oracle $O_K$, satisfying $rB_2^n
\subseteq K \subseteq RB_2^n$. Then for $\eps > 0$, a realisation of a random variable $X \in K$, having 
total
variation distance at most $\eps$ from $\unif(K)$, can be computed using
$\poly(n,\log(R/r),\log(1/\eps))$ arithmetic operations and calls to the membership oracle.   
\end{theorem}

%
%
%


\section{Bounding the Number of Crossings}
\label{sec:voronoi-proofs}

In this section, we prove bounds on the number of crossings the randomized straight line algorithm
induces on the tiling boundary $\lat + \partial\V$. For a target $\vec{t}$, starting point $\vec{x}
\in \lat$, and perturbation $Z \sim \unif(V)$, we need to bound the expected number of crossings in
phases B and C, that is
\[
(B)~ \E[|(\lat + \partial\V) \cap [\vec{x}+Z,\vec{t}+Z]|] \quad \quad 
(C)~ \E[|(\lat + \partial\V) \cap [\vec{t}+Z,\vec{t})|] \text{ .}
\] 
The phase B bound is given in Section~\ref{sec:phase-b}, and the phase C is given in
Section~\ref{sec:phase-c}. 

\subsection{Phase B estimates}
\label{sec:phase-b}

The high level idea of the phase B bound is as follows. To count the number of crossings, we break
the segment $[\vec{x}+Z,\vec{y}+Z]$ into $k$ equal chunks (we will let $k \rightarrow \infty$), and
simply count the number of chunks which cross at least $1$ boundary. By our choice of perturbation,
we can show that each point on the segment $[\vec{x}+Z,\vec{y}+Z]$ is uniformly distributed modulo
the lattice, and hence the crossing probability will be indentical on each chunk. In particular, we
will get that each crossing probability is exactly the probability that $Z$ ``escapes'' from $\V$
after moving by $(\vec{y}-\vec{x})/k$. This measures how close $Z$ tends to be to the boundary of
$\V$, and hence corresponds to a certain ``directional'' surface area to volume ratio. 

In the next lemma, we show that the escape probability is reasonably small for any symmetric convex
body, when the size of the shift is measured using the norm induced by the body. We shall use this
to prove the full phase B crossing bound in Theorem~\ref{thm:phase-b}.

\begin{lemma} Let $K \subseteq \R^n$ be a centrally symmetric convex body. Then for $Z \sim \unif(K)$
and $\vec{y} \in \R^n$, we have that
\[
\lim_{\eps \rightarrow 0} \Pr[Z+ \eps \vec{y} \notin K]/\eps \leq (n/2)\|\vec{y}\|_K
\]  
\label{lem:exit-bnd}
\end{lemma}
\begin{proof}
By applying a linear transformation to $K$ and $\vec{y}$, we may assume that $\vec{y} = \vec{e}_n$. Let $\pi_{n-1}: \R^n \rightarrow \R^{n-1}$ denote the
projection onto the first $n-1$ coordinates. Define $l:\pi_{n-1}(K) \rightarrow \R_+$ as $l(\vec{x}) =
\vol_1(\set{(\vec{x},x_n): x_n \in \R, (\vec{x},x_n) \in K})$, i.e.~the length of the chord of $K$ passing through $(\vec{x},0)$ in direction $\vec{e}_n$.

For $\vec{x} \in \pi_{n-1}(K)$, let $\set{(\vec{x},x_n): x_n \in \R, (\vec{x},x_n) \in K} = [(\vec{x},a),(\vec{x},b)]$, $a \leq b$, denote its associated
chord, where we note that $|b-a| = l(\vec{x})$. From here, conditioned on $Z$ landing on this chord, 
note that $Z+\eps \vec{e}_n \notin K$ if and only
if $Z$ lies in the half open line segment $((\vec{x},b-\eps),(\vec{x},b)]$. Given this, we have that
\begin{align}
\label{eq:eb-1}
\begin{split}
\lim_{\eps \rightarrow 0} \Pr[Z + \eps \vec{e}_n \notin K]/\eps 
													&= \lim_{\eps \rightarrow 0} (1/\eps) \int_{\pi_{n-1}(K)} \min \set{\eps, l(\vec{x})} \frac{\d\vec{x}}{\vol_n(K)} \\
													&= \lim_{\eps \rightarrow 0} \int_{\pi_{n-1}(K)} \min \set{1, l(\vec{x})/\eps} \frac{\d\vec{x}}{\vol_n(K)} \\
													&= \int_{\pi_{n-1}(K)} \frac{\d\vec{x}}{\vol_n(K)} = \frac{\vol_{n-1}(\pi_{n-1}(K))}{\vol_n(K)} \text{ .}
\end{split}
\end{align}
Let $s = 1/\|\vec{e}_n\|_K$. Since $K$ is centrally symmetric, $\R \vec{e}_n \cap K = [-s\vec{e}_n,s\vec{e}_n]$ and hence $l(\vec{0}) = 2s$.  Note
that by central symmetry of $K$, for all $\vec{x} \in \pi_{n-1}(K)$, $l(\vec{x}) = l(-\vec{x})$. Since $K$ is convex, the function $l$ is concave on
$\pi_{n-1}(K)$, and hence
\[
\max_{\vec{x} \in \pi_{n-1}(K)} l(\vec{x}) = \max_{\vec{x} \in \pi_{n-1}(K)} \frac{1}{2}l(\vec{x}) + \frac{1}{2}l(-\vec{x}) 
                                           \leq \max_{\vec{x} \in \pi_{n-1}(K)} l(\vec{0}) = 2s \text{.}
\]

Let $K' = \set{(\vec{x},x_n): \vec{x} \in \pi_{n-1}(K), 0 \leq x_n \leq l(\vec{x})}$. By concavity
of $l$, it is easy to see that $K'$ is also a convex set. Furthermore, note that $K'$ has exactly
that same chord lengths as $K$ in direction $\vec{e}_n$, and hence $\vol_n(K') = \vol_n(K)$. For $a
\in \R$, let $K'_a = K' \cap \set{(\vec{x},a): \vec{x} \in \R^{n-1}}$. Here $\R \vec{e}_n \cap K' =
[\vec{0},2s\vec{e}_n]$, and hence $K_a \neq \emptyset$, $\forall a \in [0,2s]$. Therefore $K_a =
\emptyset$ for $a > 2s$, since the maximum chord length is $l(\vec{0}) = 2s$, as well as for $a <
0$. By construction of $K'$, we see that $K'_0 = \pi_{n-1}(K) \times \set{0}$, and hence
$\vol_{n-1}(K'_0) = \vol_{n-1}(\pi_{n-1}(K))$.

Given~(\ref{eq:eb-1}), to prove the Lemma, it now suffices to show that 
\[
\frac{\vol_{n-1}(\pi_{n-1}(K))}{\vol_n(K)} = \frac{\vol_{n-1}(\pi_{n-1}(K))}{\vol_n(K')} \leq (n/2)\|\vec{e}_n\|_K \text{ .}
\]
For $a \in [0,2s]$, by convexity of $K'$ and the Brunn-Minkowski inequality on $\R^{n-1}$, we have that
\begin{align*}
\vol_{n-1}(K'_a)^{\frac{1}{n-1}} &\geq \vol_{n-1}((1-\frac{a}{2s})K'_0 + \frac{a}{2s}K'_{2s})^{\frac{1}{n-1}} \\
                                 &\geq (1-\frac{a}{2s})\vol_{n-1}(K'_0)^{\frac{1}{n-1}} + \frac{a}{2s}\vol_{n-1}(K'_{2s})^{\frac{1}{n-1}} \\
                                 &\geq (1-\frac{a}{2s})\vol_{n-1}(\pi_{n-1}(K))^{\frac{1}{n-1}} \text{ .}
\end{align*}

Therefore, we have that
\begin{align*}
\vol_n(K') &= \int_0^{2s} \vol_{n-1}(K'_a) {\rm da} \geq \vol_{n-1}(\pi_{n-1}(K)) \int_0^{2s} \left(1-\frac{a}{2s}\right)^{n-1} {\rm da} \\
           &= \vol_{n-1}(\pi_{n-1}(K))(2s) \int_0^1 (1-a)^{n-1}{\rm da} = \vol_{n-1}(\pi_{n-1}(K))(2s)/n \\
           &= \frac{2\vol_{n-1}(\pi_{n-1}(K))}{n\|\vec{e}\|_n} \text{, }
\end{align*}
as needed.
\end{proof}

\subsubsection{Proof of Theorem~\ref{thm:phase-b} (Phase B crossing bound)}

\begin{proof}
Note first that the sets $(\lat + \partial\V) \cap [\vec{x}+Z,\vec{y}+Z]$ and $(\lat + \partial\V) \cap
[\vec{x}+Z,\vec{y}+Z)$ agree unless $\vec{y}+Z \in \lat + \partial\V$.  Given that this event happens with probability
$0$ (as $\lat+\partial\V$ has $n$ dimensional Lebesgue measure $0$), we get that
\[
\E_Z[|(\lat + \partial\V) \cap [\vec{x}+Z,\vec{y}+Z]|] 
           = \E_Z[|(\lat + \partial\V) \cap [\vec{x}+Z,\vec{y}+Z)|] \text{ .}
\]
We now bound the expectation on the right hand side. For $s \in [0,1]$, define the random variable $\ell(s) = (1-s)
\vec{x} + s \vec{y} + Z$.  Let $A^k_j$, $0 \leq j < 2^k$, denote the event that 
\[
|(\lat + \partial\V) \cap [\ell(j/2^k),\ell((j+1)/2^k))| \geq 1 \quad \Leftrightarrow \quad 
|(\lat + \partial\V) \cap [\ell(j/2^k),\ell(j/2^k)+(\vec{y}-\vec{x})/2^k)| \geq 1 \text{ .}
\]
Clearly, we have that
\[
|(\lat + \partial\V) \cap [\ell(0),\ell(1))| = \lim_{k \rightarrow \infty} \sum_{j=0}^{2^k-1} A^k_j \text{ .}
\]
By the monotone convergence theorem, we get that
\begin{equation}
\label{eq:phaseb-1}
\E_Z[|(\lat + \partial\V) \cap [\ell(0),\ell(1))|] = \lim_{k \rightarrow \infty} 
\sum_{j=0}^{2^k-1} \Pr[A^k_j=1]\text{ .} 
\end{equation}
Since $\lat + \partial\V$ is by definition invariant under lattice shifts, we see that
$\Pr[A^k_j=1]$ depends only on the distribution of $\ell(j/2^k) \imod{\lat}$. Given that $Z
\imod{\lat} \sim {\rm uniform}(\R^n / \lat)$ and that $\R^n / \lat$ is shift invariant, we have that
$\ell(j/2^k) \imod{\lat} \sim {\rm uniform}(\R^n / \lat)$. In particular, this implies that
$\Pr[A^k_0] = \cdots = \Pr[A^k_{2^k}]$, and hence by Lemma~\ref{lem:exit-bnd}
\begin{align}
\begin{split}
\label{eq:phaseb-2}
\lim_{k \rightarrow \infty} \sum_{j=0}^{2^k-1} \Pr[A^k_j=1] &= \lim_{k \rightarrow \infty} 2^k \Pr[A^k_0=1] \\
                  &= \lim_{k \rightarrow \infty} 2^k \Pr[Z+(\vec{y}-\vec{x})/2^k \notin \V] \leq (n/2)\|\vec{y}-\vec{x}\|_\V \text{ ,}
\end{split}
\end{align}
as needed. The result follows by combining \eqref{eq:phaseb-1} and \eqref{eq:phaseb-2}.
\end{proof}

\subsection{Phase C estimates}
\label{sec:phase-c}

As mentioned previously in the paper, our techniques will not be sufficient to fully bound
the number of phase C crossings. However, we will use be able to give bounds for a truncation
of the phase C path, that is for $\alpha \in (0,1]$, we will bound
\[
\E[|(\lat+\partial\V) \cap [\vec{t}+Z,\vec{t}+\alpha Z]|] \text{ .}
\] 
We will give a bound of $O(n \ln 1/\alpha)$ for the above crossings in Theorem~\ref{thm:phase-c}.

For the proof strategy, we follow the approach as phase B in terms of bounding the crossing
probability on infinitessimal chunks of $[\vec{t}+Z,\vec{t}+\alpha Z]$. However, the implementation
of this strategy will be completely different here, since the points along the segment no longer have
the same distribution modulo the lattice. In particular, as $\alpha \rightarrow 0$, the
distributions get more concentrated, and hence we lose the effects of the randomness. This loss
will be surprisingly mild however, as evidenced by the claimed $\ln(1/\alpha)$ dependence.

For the infinitessimal probabilities, it will be convenient to parametrize the segment $[\vec{t}+Z,
\vec{t}+\alpha Z]$ differently than in phase B. In particular, we use $\vec{t}+Z/s$, for $s \in
[1,1/\alpha]$. From here, note that
\begin{equation}
\label{eq:phasec-inf-cross}
\Pr[(\lat+\partial\V) \cap [\vec{t}+Z/s,\vec{t}+Z/(s+\eps)] \neq \emptyset]
= \Pr[Z \in \cup_{\gamma \in [s,s+\eps]} \gamma(\lat-\vec{t}+\partial\V)] \text{ .}
\end{equation}
Taking the limit as $\eps \rightarrow 0$, we express the infinitessimal probability as a certain
weighted surface area integral over $s(\lat-\vec{t}+\partial\V)$ (see
Lemma~\ref{thm:laplace-cross-bnd}). 

In the same spirit as phase B, we will attempt to relate the surface integral to a nicely bounded
integral over all of space. To help us in this task, we will rely on a technical trick to ``smooth
out'' the distribution of $Z$. More precisely, we will replace the perturbation $Z \sim \unif(\V)$
by the perturbation $X \sim \laplace(\V,\theta)$, for an appropriate choice of $\theta$. For the
relationship between both types of perturbatoins, we will use the representation of $X$ as $rZ$, where $r
\sim \Gamma(n+1,\theta)$. We will choose $\theta$ so that $r$ is concentrated in the interval
$[1,1+\frac{1}{n}]$, which will insure that the number of crossings for $X$ and $Z$ are roughly the
same.  

The benefit of the Laplace perturbation for us will be as follows. Since it varies much more
smoothly than the uniform distribution (which has ``sharp boundaries''), it will allow us to make
the analysis of the surface integral entirely local by using the periodic structure of
$s(\lat-\vec{t}+\partial\V)$. In particular, we will be able to relate the surface integral over
each tile $s(\vec{y}-\vec{t}+\partial\V)$, $\vec{y} \in \lat$, to a specific integral over each cone
$s(\vec{y}+\conv(\vec{0},F_\vec{v}))$, $\forall \vec{v} \in \VR$, making up the tile. Under the
uniform distribution, the probability density over each tile can be challenging to analyze, since
the tile may only be partially contained in $\V$. However, under the Laplace distribution, we know
that over $s(\vec{y}+\V)$ the density can vary by at most $e^{\pm s/\theta}$ (see
Equation~\ref{lem:laplace-smooth} in the Preliminaries).  

The integral over $\R^n$ we end up using to control the surface integral over
$s(\lat-\vec{t}+\partial\V)$ turns out the be rather natural. At all scales, we simply use the
integral representation of $\E[\|X\|_\V] = n\theta$ (see Lemma~\ref{lem:surface-area-bnd}). In
particular, as $s \rightarrow \infty$, for the appropriate choice of $\theta$, this will allow us to
bound the surface integral over $s(\lat-\vec{t}+\partial\V)$ by $O(n/s)$. Integrating this bound
from $1$ to $1/\alpha$ yields the claimed $O(n \ln(1/\alpha))$ bound on the number of crossings.

This section is organized as follows. In subsection~\ref{sec:unif-to-lap}, we relate the number of
crossings for uniform and Laplace perturbations. In subsection~\ref{sec:lap-cross}, we bound the
number of crossings for Laplace perturbations. Lastly, in subsection~\ref{sec:phase-c-final-bnd}, we
combine the estimates from the previous subsections to give the full phase C in
Theorem~\ref{thm:phase-c}.

\subsubsection{Converting Uniform Perturbations to Laplace Perturbations}
\label{sec:unif-to-lap}

In this section, we show that the number of crossings induced by uniform perturbations can be
controlled by the number of crossings induced by Laplace perturbation.

We define $\theta_n = \frac{1}{(n+1)-\sqrt{2(n+1)}}$, $\gamma_n =
(1+\frac{2\sqrt{2}}{\sqrt{n+1}-\sqrt{2}})^{-1}$ for use in the rest of this section.

The following Lemma shows the $\Gamma(n+1,\theta)$ distribution is concentrated in a small interval
above $1$ for the appropriate choice of $\theta$. This will be used in Lemma~\ref{lem:unif-to-lap}
to relate the number of crossings between the uniform and Laplace perturbations.

\begin{lemma} For $r \sim \Gamma(n+1,\theta_n)$, $n \geq 2$, we have that
\[
\Pr[r \in [1,1 + \frac{2\sqrt{2}}{\sqrt{n+1}-\sqrt{2}}]] \geq \frac{1}{2}
\]
\label{lem:gamma-conc}
\end{lemma}
\begin{proof}
Remember that $\E[r] = (n+1) \theta_n$ and that $\VAR[r] = (n+1)\theta_n^2$. Letting $\sigma =
\sqrt{\VAR[r]}$, by Chebyshev's inequality
\[
\Pr[|r-\E[r]| \geq \sqrt{2}\sigma] \leq \frac{\VAR[r]}{2\sigma^2} = \frac{1}{2}
\]
The result now follows from the identities
\begin{align*}
\E[r]-\sqrt{2}\sigma &= ((n+1)-\sqrt{2(n+1)})\theta_n = 1 \\ 
\E[r]+\sqrt{2}\sigma &= ((n+1)+\sqrt{2(n+1)})\theta_n = 1 + \frac{2\sqrt{2}}{\sqrt{n+1}-\sqrt{2}}
\end{align*}
\end{proof} 

\begin{lemma} Let $\lat$ be an $n$-dimensional lattice, $n \geq 2$, and $\vec{t} \in \R^n$. Then for
$\alpha \in [0,1]$, $Z \sim \unif(\V)$ and $X \sim \laplace(\V,\theta_n)$, we have that
\[
\E_Z[|(\lat + \partial\V) \cap [Z+\vec{t}, \alpha Z + \vec{t}]|] \leq
2\E_X[|(\lat + \partial\V) \cap [X+\vec{t}, \gamma_n \alpha X + \vec{t}]]
\] 
where $\gamma_n = (1+\frac{2\sqrt{2}}{\sqrt{n+1}-\sqrt{2}})^{-1}$. 
\label{lem:unif-to-lap}
\end{lemma}
\begin{proof}
We shall use the fact that $X$ is identically distributed to $rZ$ where $r \sim
\Gamma(n+1,\theta_n)$ is sampled independently from $Z$. Conditioned on any value of $Z$, 
the following inclusion holds 
\[
[Z + \vec{t},\alpha Z + \vec{t}) \subseteq [rZ + \vec{t}, \gamma_n \alpha r Z + \vec{t})
\]
as long as $r \in [1,\gamma_n^{-1}] = [1, 1+\frac{2\sqrt{2}}{\sqrt{n+1}-\sqrt{2}}]$. 
By Lemma \ref{lem:gamma-conc}, we get that 
\begin{align*}
\E_X[|(\lat + \partial\V) \cap [X+\vec{t}, \gamma_n \alpha X + \vec{t}]|] 
&= \E_Z[~ \E_r[|(\lat + \partial\V) \cap [rZ+\vec{t}, \gamma_n \alpha rZ + \vec{t}]|] ~] \\ 
&\geq \E_Z[|(\lat + \partial\V) \cap [Z+\vec{t}, \alpha Z + \vec{t}]| \Pr[r \in [1, \gamma_n^{-1}]]] \\
&\geq \frac{1}{2} \E_Z[|(\lat + \partial\V) \cap [Z+\vec{t}, \alpha Z + \vec{t}]|] \text{,}
\end{align*}
as needed.
\end{proof}

\subsubsection{Bounding the Number of Crossing for Laplace Perturbations}
\label{sec:lap-cross}

In this section, we bound the number of crossings induced by Laplace perturbations.  The expression
for the infinitessimal crossing probabilities is given in Lemma~\ref{lem:cross-expr}, the bound
on the surface area integral over $s(\lat-\vec{t}+\partial\V)$ to $\E[\|X\|_\V]$ is given in
Lemma~\ref{lem:surface-area-bnd}, and the full phase C Laplace crossing bound is given in
Theorem~\ref{thm:laplace-cross-bnd}. 

For $\vec{t} \in \R^n$, and $s > 0$, the set $s(\lat-\vec{t}+\partial\V)$ is a shifted and scaled
version of the tiling boundary $\lat+\partial\V$. For $\vec{y} \in \lat$, and $\vec{v} \in \VR$, we
will call $s(\vec{y}-\vec{t}+F_{\vec{v}})$ a facet of $s(\lat-\vec{t}+\partial\V)$.

\begin{definition}[Tiling boundary normals]
We define the function $\eta: (\lat-\vec{t}+\partial\V) \rightarrow S^{n-1}$ as follows. For
$\vec{x} \in (\lat-\vec{t}+\partial\V)$, choose the lexicographically minimal $\vec{v} \in \VR$ such
that $\exists \vec{y} \in \lat$ satisfying $\vec{x} \in (\vec{y}-\vec{t}+F_{\vec{v}})$. Finally,
define $\eta(\vec{x}) = \vec{v}/\|\vec{v}\|_2$. 
\label{def:til-bnd-nor}
\end{definition}

Note that for $\vec{x} \in s(\lat-\vec{t}+\partial\V)$, $\eta(\vec{x}/s)$ is a unit normal to a
facet of $s(\lat-\vec{t}+\partial\V)$ containing $\vec{x}$. Furthermore, the subset of points in
$s(\lat-\vec{t}+F_{\vec{v}})$ having more than one containing facet has $n-1$ dimensional measure
$0$, and hence can be ignored from the perspective of integration over $s(\lat-\vec{t}+\partial\V)$.

The following lemma gives the expression for the infinitessimal crossing probabilities.

\begin{lemma} For $\alpha \in (0,1]$, and $X \sim \laplace(\V,\theta)$, we have that
\[
\E[|(\lat + \partial\V) \cap [X+\vec{t}, \alpha X + \vec{t}]|] = 
\int_1^{1/\alpha} \int_{s(\lat-\vec{t}+\partial\V)} 
     |\pr{\eta(\vec{x}/s)}{\vec{x}/s}|f_\V^\theta(\vec{x})  \d\vol_{n-1}(\vec{x}) {\rm ds} \text{ .}   
\]
\label{lem:cross-expr}
\end{lemma}
\begin{proof}
Firstly, shifting by $-\vec{t}$ on both sides, we get that
\[
\E[|(\lat + \partial\V) \cap [X+\vec{t}, \alpha X + \vec{t}]|] =  
\E[|(\lat -\vec{t} + \partial\V) \cap [X, \alpha X]|]\text{.}
\]

From here, we first decompose the expected number of intersections by summing over all facets.
This yields
\begin{equation}
\label{eq:ce-1}
\E[|(\lat -\vec{t} + \partial\V) \cap [X, \alpha X]|] = 
\frac{1}{2} \sum_{\vec{y} \in \lat, \vec{v} \in \VR} \E[|(\vec{y} - \vec{t} + F_{\vec{v}}) \cap [X, \alpha X]|] \text{.}
\end{equation}
The factor $1/2$ above accounts for the fact that we count each facet twice, i.e.~$\vec{y}-\vec{t}+F_\vec{v}$ and
$(\vec{y}+\vec{v})-\vec{t}+F_{-\vec{v}}$. Secondly, note that the intersections we count more than twice in the above
decomposition correspond to a countable number of lines passing through at most $n-2$ dimensional faces, and hence have
$n$ dimensional Lebesgue measure $0$. The equality in Equation~\eqref{eq:ce-1} thus follows.  

If we restrict to one facet $\vec{y}-\vec{t}+F_\vec{v}$, for some $\vec{y} \in \lat$, $\vec{v} \in \VR$, we note that
the line segment $[X,\alpha X]$ crosses the facet $\vec{y}-\vec{t}+F_\vec{v}$ at most once with probability $1$.
Hence, we get that
\begin{align}
\label{eq:ce-2}
\begin{split}
\E[|(\vec{y} - \vec{t} + F_{\vec{v}}) \cap [X, \alpha X]|] 
 &= \Pr[(\vec{y} - \vec{t} + F_{\vec{v}}) \cap [X, \alpha X] \neq \emptyset] \\
 &= \Pr[X \in \cup_{s \in [1,\frac{1}{\alpha}]} ~s(\vec{y} - \vec{t} + F_{\vec{v}})] \\
 &= \int_{\cup_{s \in [1,\frac{1}{\alpha}]} s(\vec{y} - \vec{t} + F_{\vec{v}})} 
    f_\V^\theta(\vec{x}) \d\vec{x} \text{ .}
\end{split}
\end{align}

Let $r = \pr{\vec{v}/\|\vec{v}\|_2}{\vec{y} - \vec{t} + F_{\vec{v}}}$, noting the inner product with $\vec{v}$ (and hence
$\hat{\vec{v}}$) is constant over $F_{\vec{v}}$. By possibly switching $\vec{v}$ to $-\vec{v}$ and $\vec{y}$ to
$\vec{y}+\vec{v}$ (which maintains the facet), we may assume that $r \geq 0$. Notice that by construction, for any
$\vec{x}$ in the (relative) interior of $s(\vec{y} - \vec{t} + F_{\vec{v}})$, we get that $r =
|\pr{\eta(\vec{x}/s)}{\vec{x}/s}|$, since then there is a unique facet of $s(\lat-\vec{t}+\partial\V)$ containing
$\vec{x}$. Integrating first in the in the direction $\vec{v}$, we get that
\begin{align}
\label{eq:ce-3}
\begin{split}
& \int_{\cup_{s \in [1,\frac{1}{\alpha}]} s(\vec{y} - \vec{t} + F_{\vec{v}})} 
    f_\V^\theta(\vec{x}) \d\vec{x} = 
 \int_r^{r/\alpha} \int_{(s/r)(\vec{y} - \vec{t} + F_{\vec{v}})} f_\V^\theta(\vec{x})\d\vol_{n-1}(\vec{x}){\rm ds} = \\
& \int_1^{1/\alpha} \int_{s(\vec{y} - \vec{t} + F_{\vec{v}})} rf_\V^\theta(\vec{x})\d\vol_{n-1}(\vec{x}){\rm ds} = 
 \int_1^{1/\alpha} \int_{s(\vec{y} - \vec{t} + F_{\vec{v}})} 
                    |\pr{\eta(\vec{x}/s)}{\vec{x}/s}|f_\V^\theta(\vec{x})\d\vol_{n-1}(\vec{x}){\rm ds} \text{ .}
\end{split}
\end{align}
Note that we use the $n-1$ dimensional Lebesgue measure to integrate over $s(\vec{y} - \vec{t} + F_{\vec{v}})$ since it
is embedded in $\R^n$. If $r = 0$, note that the set $\cup_{s \in [1,\frac{1}{\alpha}]} s(\vec{y} - \vec{t} +
F_{\vec{v}})$ is $n-1$ dimensional and hence has measure (and probability) $0$. This is still satisfied by the last
expression in~\eqref{eq:ce-3}, and hence the identity is still valid in this case. 

Putting everything together, combining Equation~\eqref{eq:ce-1},\eqref{eq:ce-3}, we get that
\begin{align*}
\E[|(\lat -\vec{t} + \partial\V) \cap [X, \alpha X]|] &= 
\frac{1}{2} \sum_{\vec{y} \in \lat, \vec{v} \in \VR} \int_1^{1/\alpha} \int_{s(\vec{y} - \vec{t} + F_{\vec{v}})} 
                    |\pr{\eta(\vec{x}/s)}{\vec{x}/s}|f_\V^\theta(\vec{x})\d\vol_{n-1}(\vec{x}){\rm ds} = \\
 &= \int_1^{1/\alpha} \int_{s(\lat-\vec{t}+\partial\V)} 
     |\pr{\eta(\vec{x}/s)}{\vec{x}/s}|f_\V^\theta(\vec{x})  \d\vol_{n-1}(\vec{x}) {\rm ds} \text{ ,}   
\end{align*}
as needed.
\end{proof}

The lower bound given in the following Lemma will be needed in the proof of
Lemma~\ref{lem:surface-area-bnd}.

\begin{lemma} 
For $a,b,c,d \in \R$, $c \leq d$, we have that
\[
\int_c^d |a+bh|{\rm dh} \geq (\sqrt{2}-1)(d-c) \max \set{|(a+bc)|,|a+bd|} \text{.}
\]
\vspace{-2em}
\label{lem:simple-int-bnd}
\end{lemma}
\begin{proof}
Firstly, we note that
\[
\int_c^d |a + bh|{\rm dh} = (d-c) \int_0^1 |a + b(c+(d-c)h)|{\rm dh}
                          = (d-c) \int_0^1 |(a+bc) + b(d-c)h|{\rm dh} \text{ ,}
\]
hence it suffices to prove the inequality when $c=0,d=1$. After this reduction, by possibly applying
the change of variables $h \leftarrow 1-h$, we may assume that $|a| \geq |a+b|$. Next, by changing
the signs of $a,b$, we may assume that $a \geq 0$. Hence, it remains to prove the inequality
\begin{equation}
\label{eq:sit-1}
\int_0^1 |a + bh|{\rm dh} \geq a (\sqrt{2}-1)  
\end{equation}
under the assumption that $a \geq |a+b|$, or equivalently $a \geq 0$ and $-2a \leq b \leq 0$.
Notice that if $a = 0$ or $b = 0$, the above inequality is trivially true. If $a,b \neq 0$, then dividing
inequality~\ref{eq:sit-1} by $a$, we reduce to the case where $a = 1$, $-2 \leq b < 0$. Letting
$\alpha = -1/b$, we have that $\alpha \in [1/2,\infty)$. From here, we get that
\[
\int_0^1 |1 + hb|{\rm dh} = \int_0^1 |1 - h/\alpha|{\rm dh} = (1/2)(\alpha + (1-\alpha)^2/\alpha)
\]
The derivative of the above expression is $1-\frac{1}{2\alpha^2}$. The expression is thus minimized for
$\alpha = \frac{1}{\sqrt{2}} > 1/2$, and the result follows by plugging in this value.  
\end{proof}

We now prove the bound on the surface integral in terms of the expectation $\E[\|X\|_\V]$.

\begin{lemma} For $s \geq 1$ and $X \sim \laplace(\V,\theta)$, we have that
\begin{align*}
\int_{s(\lat-\vec{t}+\partial\V)} |\pr{\eta(\vec{x}/s)}{\vec{x}/s}|f_\V^\theta(\vec{x})\d\vol_{n-1}(\vec{x})
\leq c \max \set{\frac{n}{s^2}, \frac{1}{\theta s}} ~ \E[\|X\|_\V] 
= c \max \set{\frac{n^2\theta}{s^2},\frac{n}{s}} \text{,}
\end{align*}
for $c = \frac{e^2}{2(\sqrt{2}-1)} \leq 9$.
\label{lem:surface-area-bnd}
\end{lemma}
\begin{proof}
We first prove the equality on the right hand side. We remember that $X$ is identically distributed
to $rZ$ where $r \sim \Gamma(n+1,\theta)$ and $Z \sim \unif(\V)$. From here, we have that
\begin{align*}
\E[\|X\|_\V] &= \E[\|rZ\|_\V] = \E[r] \E[\|Z\|_\V] \\
             &= (n+1)\theta \int_0^1 \Pr[\|Z\|_\V \geq s] {\rm ds} \\
             &= (n+1)\theta \int_0^1 (1-s^n) {\rm ds} = (n+1)\theta\left(\frac{n}{n+1}\right) = n \theta \text{ ,}
\end{align*}
as needed.

We now prove the first inequality. To prove the bound, we write the integral expressing
$\E[\|X\|_\V]$ over the cells of $s(\lat-\vec{t}+\partial\V)$, and compare the integral over each
cell to the corresponding boundary integral. To begin
\begin{align}
\label{eq:sab-start}
\begin{split}
\E[\|X\|_\V] &= \int_{\R^n} \|\vec{x}\|_\V f_\V^\theta(\vec{x}) \d\vec{x} 
		= \sum_{\vec{y} \in \lat} 
			\int_{s(\vec{y}-\vec{t}) + s\V} \|\vec{x}\|_\V f_\V^\theta(\vec{x}) \d\vec{x} \\
	 &= \sum_{\vec{y} \in \lat, \vec{v} \in \VR} 
			\int_{s(\vec{y}-\vec{t}) + \conv(\vec{0}, sF_{\vec{v}})} \|\vec{x}\|_\V f_\V^\theta(\vec{x}) \d\vec{x} \text{ .}
\end{split}
\end{align}
Fix $\vec{y} \in \lat$ and $\vec{v} \in \V$ in the above sum.  Noting that
$\pr{\vec{v}/\|\vec{v}\|_2}{sF_{\vec{v}}} = s\|\vec{v}\|_2/2$ by construction, and integrating first in
the direction $\vec{v}$, we get that
\begin{align}
\label{eq:sab-1}
\begin{split}
&\int_{s(\vec{y}-\vec{t})+\conv(\vec{0},sF_{\vec{v}})} \|\vec{x}\|_\V f_\V^\theta(\vec{x}) \d\vec{x} = \\
&\int_0^{s\|\vec{v}\|_2/2} \int_{\frac{2h}{s\|\vec{v}\|_2}(sF_{\vec{v}})} 
      \|s(\vec{y}-\vec{t})+\vec{x}\|_\V f_\V^\theta(s(\vec{y}-\vec{t})+\vec{x}) \d\vol_{n-1}(\vec{x}){\rm dh} \text{ .}
\end{split}
\end{align}
In the above, we use the $n-1$ dimensional Lebesgue measure to integrate over
$\frac{2h}{s\|\vec{v}\|_2}F_{\vec{v}}$ since it is embedded in $\R^n$ (we also do this for ease of
notation). Setting $\beta = \frac{2h}{s\|\vec{v}\|_2}$, note that $\beta \in [0,1]$. In
equation~\eqref{eq:sab-1}, $\beta$ represents the convex combination between $\vec{0}$ and
$sF_{\vec{v}}$, that is $\conv(\vec{0},F_{\vec{v}}) = \bigcup_{\beta \in [0,1]} \beta sF_{\vec{v}}$.
Performing a change of variables, Equation~\eqref{eq:sab-1} simplifies to
\begin{align}
\label{eq:sab-2}
\begin{split}
&\int_0^1 \int_{h (sF_{\vec{v}})} (s\|\vec{v}\|_2/2)
		\|s(\vec{y}-\vec{t})+\vec{x}\|_\V f_\V^\theta(s(\vec{y}-\vec{t})+\vec{x}) \d\vol_{n-1}(\vec{x}){\rm dh} = \\ 
&\int_{sF_{\vec{v}}} \int_0^1 (s\|\vec{v}\|_2/2) 
	\|s(\vec{y}-\vec{t})+h\vec{x}\|_\V f_\V^\theta(s(\vec{y}-\vec{t})+h\vec{x}) h^{n-1}{\rm dh}~\d\vol_{n-1}(\vec{x}) = \\
&\int_{sF_{\vec{v}}} \int_0^1 (s\|\vec{v}\|_2/2) 
	\|s(\vec{y}-\vec{t})+(1-h)\vec{x}\|_\V f_\V^\theta(s(\vec{y}-\vec{t})+(1-h)\vec{x}) (1-h)^{n-1}
          {\rm dh}~\d\vol_{n-1}(\vec{x}) \text{ .}
\end{split}
\end{align}
From here we note that
\begin{align}
\label{eq:sab-3}
\begin{split}
(s\|\vec{v}\|_2/2) \|s(\vec{y}-\vec{t})+(1-h)\vec{x}\|_\V
&= (s\|\vec{v}\|_2/2) \max_{\vec{w} \in \VR} 
~\left|\frac{2\pr{\vec{w}}{s(\vec{y}-\vec{t})+(1-h)\vec{x}}}{\pr{\vec{w}}{\vec{w}}}\right| \\
&\geq (s\|\vec{v}\|_2/2) ~ \left|\frac{2\pr{\vec{v}}{s(\vec{y}-\vec{t})+(1-h)\vec{x}}}{\pr{\vec{v}}{\vec{v}}}\right| \\
&= s^2 \left|\pr{\vec{v}/\|\vec{v}\|_2}{(\vec{y}-\vec{t})+(1-h)\vec{x}/s}\right| \text{.}
\end{split}
\end{align}
From inequality~\eqref{eq:sab-3}, we have that the expression in equation~\eqref{eq:sab-2} is greater than or equal to
\begin{equation}
\label{eq:sab-4}
\int_{sF_{\vec{v}}} s^2 \int_0^1 \left|\pr{\vec{v}/\|\vec{v}\|_2}{(\vec{y}-\vec{t})+(1-h)\vec{x}/s}\right| 
       f_\V^\theta(s(\vec{y}-\vec{t})+(1-h)\vec{x}) (1-h)^{n-1}{\rm dh}~\d\vol_{n-1}(\vec{x}) \text{ .}
\end{equation}
To compare to the surface integral, we now lower bound the inner integral.
\begin{claim} For $\|\vec{x}\|_\V \leq s$, we have that
\begin{align*}
&\int_0^1 \left|\pr{\vec{v}/\|\vec{v}\|_2}{(\vec{y}-\vec{t})+(1-h)\vec{x}/s}\right|
       f_\V^\theta(s(\vec{y}-\vec{t})+(1-h)\vec{x}) (1-h)^{n-1}{\rm dh} \\
&\geq ~ e^{-2}(\sqrt{2}-1) \min \set{\frac{1}{n},\frac{\theta}{s}} ~
|\pr{\vec{v}/\|\vec{v}\|_2}{(\vec{y}-\vec{t})+\vec{x}/s}| f_\V^\theta(s(\vec{y}-\vec{t})+\vec{x})
\end{align*}
\label{cl:sab-inner}
\end{claim}
\begin{proof}
Note that for $0 \leq h \leq \min \set{\frac{1}{n},\frac{\theta}{s}}$, we have that
\begin{align*}
f_\V^\theta(s(\vec{y}-\vec{t})+(1-h)\vec{x})(1-h)^{n-1} 
&\geq f_\V^\theta(s(\vec{y}-\vec{t})+\vec{x})e^{-\|h\vec{x}\|_\V/\theta}(1-h)^{n-1} \\
&\geq f_\V^\theta(s(\vec{y}-\vec{t})+\vec{x})e^{-1}(1-1/n)^{n-1} \geq e^{-2}f_\V^\theta(s(\vec{y}-\vec{t})+\vec{x})
\text{ .} 
\end{align*}
Hence, using the above and Lemma~\ref{lem:simple-int-bnd}, we have that
\begin{align*}
&\int_0^1 \left|\pr{\vec{v}/\|\vec{v}\|_2}{(\vec{y}-\vec{t})+(1-h)\vec{x}/s}\right|
       f_\V^\theta(s(\vec{y}-\vec{t})+(1-h)\vec{x}) (1-h)^{n-1}{\rm dh} \\
&\geq e^{-2}f_\V^\theta(s(\vec{y}-\vec{t})+\vec{x})\int_0^{\min \set{\frac{1}{n},\frac{\theta}{s}}}
\left|\pr{\vec{v}/\|\vec{v}\|_2}{(\vec{y}-\vec{t})+(1-h)\vec{x}/s}\right|{\rm dh} \\ 
&\geq e^{-2}(\sqrt{2}-1) \min \set{\frac{1}{n},\frac{\theta}{s}} ~ 
\left|\pr{\vec{v}/\|\vec{v}\|_2}{(\vec{y}-\vec{t})+\vec{x}/s}\right|f_\V^\theta(s(\vec{y}-\vec{t})+\vec{x})\text{ ,}
\end{align*}
as needed.
\end{proof}

Given Claim~\ref{cl:sab-inner}, we get that expression~\eqref{eq:sab-4} is greater than or equal to 
\begin{align*}
& \int_{sF_{\vec{v}}} s^2 e^{-2}(\sqrt{2}-1) \min \set{\frac{1}{n},\frac{\theta}{s}} ~ 
|\pr{\vec{v}/\|\vec{v}\|_2}{(\vec{y}-\vec{t})+\vec{x}/s}| f_\V^\theta(s(\vec{y}-\vec{t})+\vec{x})\d\vol_{n-1}(\vec{x}) = \\
& e^{-2}(\sqrt{2}-1) \min \set{\frac{s^2}{n},s\theta} \int_{s(\vec{y}-\vec{t}+F_{\vec{v}})}
|\pr{\eta_s(\vec{x}/s)}{\vec{x}/s}|f_\V^\theta(\vec{x})\d\vol_{n-1}(\vec{x}) \text{ .}
\end{align*}
Putting everything together, combining the above with equation~\eqref{eq:sab-start}, we get that
\begin{align*}
\E[\|X\|_\V] &= \sum_{\vec{y} \in \lat, \vec{v} \in \VR} 
			\int_{s(\vec{y}-\vec{t}) + \conv(\vec{0}, sF_{\vec{v}})} \|\vec{x}\|_\V f_\V^\theta(\vec{x}) \d\vec{x} \\
&\geq e^{-2}(\sqrt{2}-1) \min \set{\frac{s^2}{n},s\theta} \sum_{\vec{y} \in \lat, \vec{v} \in \VR} 
\int_{s(\vec{y}-\vec{t}+F_{\vec{v}})} 
|\pr{\eta_s(\vec{x}/s)}{\vec{x}/s}|f_\V^\theta(\vec{x})\d\vol_{n-1}(\vec{x}) = \\
&= 2e^{-2}(\sqrt{2}-1) \min \set{\frac{s^2}{n},s\theta} 
\int_{s(\lat-\vec{t}+\partial\V)} 
|\pr{\eta_s(\vec{x}/s)}{\vec{x}/s}|f_\V^\theta(\vec{x})\d\vol_{n-1}(\vec{x}) \text{ ,}
\end{align*}
where the last equality follows since each facet in $s(\lat-\vec{t}+\partial\V)$ is counted twice. The lemma thus
follows.
\end{proof}

The following gives the full phase C bound for Laplace perturbations.

\begin{theorem} For $\alpha \in (0,1]$ and $X \sim \laplace(\V,\theta)$, we have that
\[
\E[|(\lat + \partial\V) \cap [X+\vec{t},\alpha X+\vec{t}]|]
\leq cn\left(\frac{n\theta}{2}\left(1-\frac{1}{s^*}\right) + \ln\left(\frac{1}{\alpha s^*}\right)\right)
\]
where $c = \frac{e^2}{2(\sqrt{2}-1)} \leq 9$ and $s^* = \max \set{1,n\theta}$.
\label{thm:laplace-cross-bnd}
\end{theorem}
\begin{proof}
Using Lemmas~\ref{lem:cross-expr} and~\ref{lem:surface-area-bnd}, we have that
\begin{align*}
\E[|(\lat + \partial\V) \cap [X+\vec{t},\alpha X+\vec{t}]|] &=
\int_1^{1/\alpha} \int_{s(\lat-\vec{t}+\partial\V)} 
     |\pr{\eta(\vec{x}/s)}{\vec{x}/s}|f_\V^\theta(\vec{x})  \d\vol_{n-1}(\vec{x}) {\rm ds} \\   
&\leq c \int_1^{1/\alpha} \max \set{\frac{n^2\theta}{s^2}, \frac{n}{s}} {\rm ds} 
= c \int_1^{s^*} \frac{n^2\theta}{s^2} {\rm ds} + c \int_{s^*}^{1/\alpha} \frac{n}{s} {\rm ds} \\
&= cn^2\theta\left(1/2-\frac{1}{s^*}\right) + cn \ln\left(\frac{1}{\alpha s^*}\right) \\ 
&= cn\left(\frac{n\theta}{2}\left(1-\frac{1}{s^*}\right) + \ln\left(\frac{1}{\alpha s^*}\right)\right) \text{ ,} 
\end{align*}
as needed.
\end{proof}

\subsubsection{Proof of Theorem~\ref{thm:phase-c} (Phase C crossing bound)}
\label{sec:phase-c-final-bnd}

\begin{proof}
We recall that $\theta_n = \frac{1}{(n+1)-\sqrt{2(n+1)}}$ and $\gamma_n = \left(1 +
\frac{2\sqrt{2}}{\sqrt{n+1}-\sqrt{2}}\right)^{-1}$. Note that $n\theta_n > 1$.
 
By Lemma~\ref{lem:unif-to-lap} and Theorem~\ref{thm:laplace-cross-bnd}, for $X \sim \laplace(\V,\theta_n)$, 
we have that
\begin{align*}
\E[|(\lat + \partial\V) \cap [Z+\vec{t},\alpha Z+\vec{t}]|] 
&\leq 2\E[|(\lat + \partial\V) \cap [X+\vec{t},\gamma_n\alpha X+\vec{t}]|] \\
&\leq 2cn\left(\frac{n\theta_n}{2}\left(1-\frac{1}{n\theta_n}\right) + 
               \ln\left(\frac{1}{\alpha\gamma_n n\theta_n}\right)\right) \\
&\leq \frac{e^2}{\sqrt{2}-1} n (2 + \ln(4/\alpha)) \text{ , for $n \geq 2$ ,}
\end{align*}
as needed.
\end{proof}

\section{Missings proofs from Section~\ref{sec:path-length}}
\label{sec:appendix-path-length}

\begin{proof}[Proof of Lemma~\ref{lem:rsl-complexity} (Randomized Straight Line Complexity)]
We recall the three phases of the randomized straight line algorithm:
\begin{enumerate}[label=(\Alph*)]
\item Move from $\vec{x}$ to $\vec{x}+Z$.
\item Follow the sequence of Voronoi cells from $\vec{x}+Z$ to $\vec{y}+Z$.
\item Follow the sequence of Voronoi cells from $\vec{y}+Z$ to $\vec{y}$.
\end{enumerate}

\paragraph{{\bf Characterizing the Path Length:}} We will show that with probability $1$, the length
of the path on $\cal{G}$ induced by the three phases is
\[
|(\lat + \partial\V) \cap [\vec{x}+Z,\vec{t}+Z]| + |(\lat + \partial\V) \cap [\vec{t}+Z,\vec{t})| \text{ .}
\]
Firstly, note that since $Z \in \V$ and $\vec{x} \in \lat$, $\vec{x}+Z$ and $\vec{x}$ lie in the
same Voronoi cell $\vec{x}+Z$, and hence phase A corresponds to the trivial path $\vec{x}$.  Hence,
we need only worry about the number of edges induced by phases B and C.

The following claim will establish the structure of a generic intersection pattern with the tiling
boundary, which will be necessary for establishing the basic properties of the path. 

\begin{claim} With probability $1$, the path $[\vec{x}+Z,\vec{t}+Z] \cup [\vec{t}+Z,\vec{t})$ only
intersects $\lat+\partial\V$ in the relative interior of its facets. Furthermore, with probability
$1$, the intersection consists of isolated points, and $\vec{x}+Z,\vec{t}+Z \notin \lat+\partial\V$.
\label{cl:gen-path}
\end{claim}
\begin{proof}
We prove the first part. Let $C_1,\dots,C_k$ denote the $n-2$ dimensional faces of $\V$. Note that
the probability of not hitting $\lat+\partial\V$ in the relative interior of it facets, can be
expressed as
\begin{align}
\label{eq:rsl-c-1}
\begin{split}
&\Pr[\left([\vec{x}+Z,\vec{t}+Z] \cup [\vec{t}+Z,\vec{t})\right) 
    \cap \left(\cup_{i \in [k]} \lat + C_i\right) \neq \emptyset] \leq \\
&\sum_{\vec{y} \in \lat, i \in [k]} \Pr[[\vec{x}+Z,\vec{t}+Z] \cap (\vec{y}+C_i) \neq \emptyset] 
                                 + \Pr[[\vec{t}+Z,\vec{t}) \cap (\vec{y}+C_i) \neq \emptyset]  \text{ .} 
\end{split}
\end{align}
Here the last inequality is valid since $\lat$ is countable.  Analyzing each term separately, we
see that 
\[
\Pr[[\vec{x}+Z,\vec{t}+Z] \cap (\vec{y}+C_i) \neq \emptyset] = \Pr[Z \in
\vec{y}+C_i-[\vec{x},\vec{t}]] = 0 \text{.}
\]
To justify the last equality, note that since $C_i$ is $n-2$ dimensional, $\vec{y}+C_i-[\vec{x},\vec{t}]$ is
at most $n-1$ dimensional (since the line segment can only add $1$ dimension). Therefore
$\vec{y}+C_i-[\vec{x},\vec{t}]$ has $n$ dimensional Lebesgue measure $0$, and in particular
probability $0$ with respect to $\unif(\V)$. Next, we have that
\[
\Pr[[\vec{t}+Z,\vec{t}) \cap (\vec{y}+C_i) \neq \emptyset] = \Pr[Z \in \cup_{s > 1}
s(\vec{y}+C_i-\vec{t})] = 0 \text{ ,}
\]
where the last equality follows since $\cup_{s > 1} s(\vec{y}+C_i-\vec{t})$ is at most $n-1$ dimensional.
Hence the probability in~\eqref{eq:rsl-c-1} is $0$, as needed.

We now prove the second part. Note that if the path $[\vec{x}+Z,\vec{t}+Z] \cup
[\vec{t}+Z,\vec{t})$ does not intersect $\lat+\partial\V$ in isolated points (i.e.~the intersection
contains a non-trivial interval), then either $[\vec{x}+Z,\vec{t}+Z]$ or $[\vec{t}+Z,\vec{t})$ must
intersect some facet of $\lat+\partial\V$ in a least $2$ distinct points. 

Let $F_\vec{v}$ be the facet of $\V$ induced by $\vec{v} \in \VR$. If $[\vec{t}+Z,\vec{t})$
intersects $\vec{y} + F_{\vec{v}}$, for some $\vec{y} \in \lat$, in two distinct points then we must
have that $\pr{\vec{v}}{Z}=0$. Since $\Pr[\cup_{\vec{v} \in \VR} \set{\pr{\vec{v}}{Z} = 0}] = 0$,
this event happens with probability $0$. Next, note that $[\vec{x}+Z,\vec{t}+Z]$ intersects
$\vec{y}+F_{\vec{v}}$ in two distinct points, if and only if $\pr{\vec{v}}{\vec{x}-\vec{t}}=0$ and
$\pr{\vec{v}}{\vec{Z}} = \pr{\vec{v}}{\vec{y}+\vec{v}/2}$. But then, the probability of this
happening for any facet can be bounded by
\[
\Pr[\cup_{\vec{y} \in \lat, \vec{v} \in \VR} \set{\pr{\vec{v}}{Z} = \pr{\vec{v}}{\vec{y}+\vec{v}/2}}] = 0
\]   
since $\lat \times \VR$ is countable. 

For the last part, note that since $\lat+\partial\V$ is the union of $n-1$ dimensional pieces,
$\Pr[\vec{x}+Z \in \lat+\partial\V] + \Pr[\vec{t}+Z \in \lat+\partial\V] = 0$.

The claim thus follows.
\end{proof}

Conditioning on the intersection structure given in claim~\ref{cl:gen-path}, we now describe the
associated path on $\cal{G}$. Let $\vec{p}_1,\dots,\vec{p}_k$ denote the points in
$([\vec{x}+Z,\vec{t}+Z] \cup [\vec{t}+Z,\vec{t})) \cap (\lat+\partial\V)$ ordered in order of
appearance on the path $[\vec{x}+Z,\vec{t}+Z] \cup [\vec{t}+Z,\vec{t})$ from left to right. Letting
$\vec{p}_{k+1} = \vec{t}$, let $\vec{y}_i \in \lat$, $1 \le i \leq k$, denote the center of the
unique Voronoi cell in $\lat+\V$ containing the interval $[\vec{p}_i,\vec{p}_{i+1}]$. Note that the
existence of $\vec{y}_i$ is guaranteed since the Voronoi cells in the tiling $\lat+\V$ are interior
disjoint, and the open segment $(\vec{p}_i,\vec{p}_{i+1})$ lies in the interior of some Voronoi cell
by convexity of the cells.  

Letting $\vec{y}_0 = \vec{x}$, we now claim that $\vec{y}_0,\vec{y}_1,\dots,\vec{y}_k$ form a valid
path in $\cal{G}$. To begin, we first establish that $\vec{y}_i \neq \vec{y}_{i+1}$, $0 \leq i \leq
k$.  Firstly, since $\vec{x}+Z \notin \lat+\partial\V$, we have that $Z$ is in the interior of $\V$,
and hence the ray starting at $Z$ in the direction of $\vec{p}_1$ exits $\vec{x}+\V$ at $\vec{p}_1$
and never returns (by convexity of $\V$). Furthermore, since $\vec{p}_1 \neq \vec{t}+Z$, the Voronoi
cell $\vec{y}_1+\V$ must contain a non-trivial interval on this ray starting at $\vec{p}_1$,
i.e.~$[\vec{p}_1,\vec{p}_2]$, and hence $\vec{y}_1 \neq \vec{x}$. Indeed, for the remaining cases,
the argument follows in the same way as long as the Voronoi cell $\vec{y}_{i+1}+\V$ contains a
non-trivial interval of the ray exiting $\vec{y}_i+\V$.  Note that this is guaranteed by the
assumption that $\vec{t}+Z \notin \partial\V$ and by the fact that none of the $\vec{p}_i$s equals
$\vec{t}$. Hence $\vec{y}_i \neq \vec{y}_{i+1}$, $0 \leq i \leq k$, as needed.

Next, note that each $p_i$, $i \in [k]$, belongs to the relative interior of some facet of
$\lat+\partial\V$. Furthermore, by construction $p_i \in \vec{y}_{i-1} + \partial\V$ and $p_i \in
\vec{y}_i + \partial\V$. Since the relative interior of facets of $\lat+\partial\V$ touch exactly
two adjacent Voronoi cells, and since $\vec{y}_{i-1} \neq \vec{y}_i$, we must have that $\vec{p}_i
\in \vec{y}_{i-1} + F_{\vec{v}}$, where $\vec{v} = \vec{y}_i-\vec{y}_{i-1} \in \VR$. Hence the path
$\vec{y}_0,\vec{y}_1,\dots,\vec{y}_k$ is valid in $\cal{G}$ as claimed.

From here, note that the length of is indeed $k = |([\vec{x}+Z,\vec{t}+Z] \cup [\vec{t}+Z,\vec{t})
\cap (\lat+\partial\V)|$. Since this holds with probability $1$, we get that the expected path
length is
\[
\E[|(\lat + \partial\V) \cap [\vec{x}+Z,\vec{t}+Z]|] + 
\E[|(\lat + \partial\V) \cap [\vec{t}+Z,\vec{t})|] \text{ .}
\]
as needed.

\paragraph{{\bf Computing the Path}:} We now explain how to compute each edge of the path
using $O(n|\VR|)$ arithmetic operations, conditioning on the conclusions of Claim~\ref{cl:gen-path}.

In constructing the path, we will in fact compute the intersection points
$\vec{p}_1,\dots,\vec{p}_k$ as above, and the lattice points $\vec{y}_1,\dots,\vec{y}_k$. As one
would expect, this computation is broken up in phase B and C, corresponding to computing the
intersection / lattice points for $[\vec{x}+Z,\vec{t}+Z]$ in phase B, followed by the intersection
/ lattice points from $[\vec{t}+Z, \vec{t})$ in Phase C.  

For each phase, we will use a generic line following procedure that given vectors $\vec{a},\vec{b}
\in \R^n$, and a starting lattice point $\vec{z} \in \lat$, such that $\vec{a} \in \vec{z} + \V$,
follows the path of Voronoi cells along the line segment $[\vec{a},\vec{b})$, and outputs a lattice
vector $\vec{w} \in \lat$ satisfying $\vec{b} \in \vec{w}+\V$. To implement phase B, we initialize
the procedure with $\vec{x}+Z,\vec{t}+Z$ and starting point $\vec{x}$. For phase C, we give it
$\vec{t}+Z,\vec{t}$ and the output of phase B as the starting point. 

We describe the line following procedure. Let $\ell(\alpha) = (1-\alpha) \vec{a} + \alpha \vec{b}$,
for $\alpha \in [0,1]$, i.e.~the parametrization of $[\vec{a}+Z,\vec{b}+Z]$ as a function of time.
The procedure will have a variable for $\alpha$, which will be set at its bounds at the beginning
and end of the procedure, starting at $0$ ending at $\geq 1$, and in intermediate steps will correspond
to an intersection point.  We will also have a variable $\vec{w} \in \lat$, corresponding to the
current Voronoi cell center.  We will maintain the invariant that $\ell(\alpha) \in \vec{w} + \V$,
and furthermore that $\ell(\alpha) \in \vec{w}+\partial\V$ for $\alpha \in (0,1)$.  

The line following algorithm is as follows:

\begin{algorithm}
\DontPrintSemicolon
\KwData{$\vec{a},\vec{b} \in \R^n$, $\vec{z} \in \lat$, $\vec{a} \in \vec{z}+\V$}
\KwResult{$\vec{w} \in \lat$ such that $\vec{b} \in \vec{w}+\V$}
$\vec{w} \leftarrow \vec{z}$, $\vec{e} \leftarrow 0$, $\alpha \leftarrow 0$\; 
$\VR' \leftarrow \set{\vec{v} \in \VR: \pr{\vec{v}}{\vec{b}-\vec{a}} > 0}$\;
\Repeat{$\alpha \geq 1$}{
  $\vec{w} \leftarrow \vec{w} + \vec{e}$\;
	$\vec{e} \leftarrow \argmin_{\vec{v} \in \VR'}
	\frac{\pr{\vec{v}}{\vec{v}/2+\vec{w}-\vec{a}}}{\pr{\vec{b}-\vec{a}}{\vec{v}}}$\; 
	$\alpha \leftarrow \frac{\pr{\vec{e}}{\vec{e}/2+\vec{w}-\vec{a}}}{\pr{\vec{b}-\vec{a}}{\vec{e}}}$\;
}
\Return{$\vec{w}$}
\end{algorithm}

Described in words, each loop iteration does the following: given the current Voronoi cell
$\vec{w}+\V$, and the entering intersection point $\ell(\alpha)$ of the line segment
$[\vec{a},\vec{b}]$ with respect to $\vec{w}+\V$, we first compute the exiting intersection point
$\ell(\alpha')$, $\alpha' > \alpha$, and the exiting facet $\vec{w}+F_{\vec{e}}$. If $\alpha' \geq
1$, we know that $\vec{b} \in [\ell(\alpha),\ell(\alpha')] \subseteq \vec{w}+\V$, and hence we may
return $\vec{w}$. Otherwise, we move to the center of the Voronoi cell sharing the facet
$\vec{w}+F_{\vec{e}}$ opposite $\vec{w}$.   

To verify the correctness, we need only show that the line $[\vec{a},\vec{b}]$ indeed exits
$\vec{w}+\V$ through the facet $\vec{w}+F_{\vec{e}}$ at the end of each iteration. Note that by our
invariant $\ell(\alpha) \in \vec{w}+\V$ at the beginning of the iteration, and hence
\begin{align}
\label{eq:rsl-c-2}
\begin{split}
\pr{\vec{v}}{\ell(\alpha)-\vec{w}} \leq \frac{1}{2}\pr{\vec{v}}{\vec{v}}, 
\quad \forall \vec{v} \in \VR \quad \Leftrightarrow \\ 
\pr{\vec{v}}{(1-\alpha)\vec{a}+\alpha \vec{b}-\vec{w}} \leq \frac{1}{2}\pr{\vec{v}}{\vec{v}}, 
\quad \forall \vec{v} \in \VR \quad \Leftrightarrow \\ 
\alpha \pr{\vec{v}}{\vec{b}-\vec{a}} \leq \pr{\vec{v}}{\vec{v}/2-\vec{a}+\vec{w}},
\quad \forall \vec{v} \in \VR
\end{split}
\end{align}

Since we move along the line segment $[\vec{a},\vec{b}]$ by increasing $\alpha$, i.e.~going from $\vec{a}$
to $\vec{b}$, note that the only constraints that can be eventually violated as we increase $\alpha$
are those for which $\pr{\vec{v}}{\vec{b}-\vec{a}} > 0$. Hence, in finding the first violated
constraint (i.e.~exiting facet), we may restrict our attention to the subset of Voronoi relevant vectors $\VR' =
\set{\vec{v} \in \VR: \pr{\vec{v}}{\vec{b}-\vec{a}} > 0}$ as done in the algorithm.

From~\eqref{eq:rsl-c-2}, we see that we do not to cross any facet $\vec{w}+F_{\vec{v}}$, $\vec{v}
\in \VR'$, as long as
\[
\alpha \leq \frac{\pr{\vec{v}}{\vec{v}/2-\vec{a}+\vec{w}}}{\pr{\vec{v}}{\vec{b}-\vec{a}}} ,
\quad \forall \vec{v} \in \VR'\text{ .}
\] 
Hence the first facet we violate must be induced by
\begin{equation}
\label{eq:rsl-c-3}
\vec{e} = \argmin_{\vec{v} \in \VR'} \frac{\pr{\vec{v}}{\vec{v}/2-\vec{a}+\vec{w}}}{\pr{\vec{v}}{\vec{b}-\vec{a}}}\text{ .}
\end{equation}

Letting $\alpha' = \frac{\pr{\vec{e}}{\vec{e}/2-\vec{a}+\vec{w}}}{\pr{\vec{e}}{\vec{b}-\vec{a}}}$,
we see that $\ell(\alpha') \in \vec{w}+F_{\vec{e}}$ is the correctly computed exiting point
(corresponding to $\ell(\alpha)$ at the end of the loop iteration), and that $\vec{w}+F_{\vec{e}}$
is the exiting facet. Since the facet $\vec{w}+F_{\vec{e}}$ is shared by $(\vec{w}+\vec{e})+\V$, we
see that $\ell(\alpha') \in (\vec{w}+\vec{e})+\partial\V$, and hence the invariant is maintained in
the next iteration. The line following algorithm is thus correct.

Notice that each iteration of the line following procedure clearly requires at most $O(n|\VR|)$
arithmetic operations. We note that the conclusions of Claim~\ref{cl:gen-path} are only needed to
ensure that each iteration of the path finding procedure can be associated with exactly one
intersection point in $([\vec{x}+Z,\vec{t}+Z] \cup [\vec{t}+Z,\vec{t})) \cap (\lat+\partial\V)$. In
particular, it assures that the minimizer in~\eqref{eq:rsl-c-3} is unique. This concludes the proof
of the Lemma.
\end{proof}

\begin{proof}[Proof of Lemma~\ref{lem:bit-bnd} (Bit length bound)]
Clearly, 
\begin{equation}
\label{eq:bb-1}
\bar{q} \leq (\prod_{ij} q^B_{ij})(\prod_i q^{\vec{t}}_i) \Rightarrow 
\log_2{\bar{q}} \leq \sum_{ij} \log_2(q^B_{ij}) + \sum_i q^{\vec{t}}_i \text{ .}
\end{equation}
Hence $\log_2 \bar{q}$ is smaller than the sum of encoding sizes of the denominators of the entries
of $B$ and $\vec{t}$. Next, it is well known that $\mu(\lat) \leq \frac{1}{2} \sqrt{\sum_{ij}
B_{ij}^2}$ (see for example~\cite{DBLP:journals/combinatorica/Babai86}). From here, we get that
\begin{align}
\label{eq:bb-2}
\begin{split}
\log_2 \mu(\lat) &\leq \log_2\left(\sqrt{\sum_{ij} B_{ij}^2}\right) \leq \log_2\left(\sqrt{\sum_{ij}
  (p^B_{ij})^2}\right) \leq \log_2\left(\sqrt{\prod_{ij} ((p^B_{ij})^2+1})\right) \\ &= \sum_{ij}
\log_2\left(\sqrt{(p^B_{ij})^2+1}\right) \leq \sum_{ij} \log_2\left(|p^B_{ij}| + 1\right)
\end{split}
\end{align}
Hence $\log_2 \mu(\lat)$ is less than the sum of encoding sizes of the numerators of the entries in
$B$. The bound $\log_2(\bar{q}\mu(\lat)) \leq \enc{B}+\enc{\vec{t}}$ now follows by
adding~\eqref{eq:bb-1},\eqref{eq:bb-2}.

We now bound $\log_2(\mu(\lat)/\lambda_1(\lat))$. Letting $\tilde{q} = \prod_{ij} q^B_{ij}$, note
that
\[
\tilde{q}\lambda_1(\lat) = \lambda_1(\tilde{q}\lat) \geq \lambda_1(\Z^n) = 1 \text{.}
\]
Therefore $1/\lambda_1(\lat) \leq \tilde{q}$. Since $\log_2 \tilde{q} \leq \sum_{ij}
\log_2(q^B_{ij})$ and $\log_2 (\mu(\lat)/\lambda_1(\lat)) \leq \log_2(\tilde{q} \mu(\lat))$,
combining with~\eqref{eq:bb-2} we get that $\log_2 (\mu(\lat)/\lambda_1(\lat)) \leq \enc{B}$ as
needed.
\end{proof}

\section{Open Problems}
    
Our work here raises a number of natural questions. Firstly, given the improvement for $\CVPP$, it
is natural to wonder whether any of the insights developed here can be used to improve the
complexity upper bound for $\CVP$. As mentioned previously, this would seem to require new
techniques, and we leave this as an open problem.

Secondly, while we now have a number of methods to navigate over the Voronoi graph, we have no lower
bounds on the lengths of the path they create. In particular, it is entirely possible that either
the ${\rm MV}$ path or the simple deterministic straight line path, also yield short paths on the
Voronoi graph. Hence, showing either strong lower bounds for these methods or new upper bounds is an
interesting open problem. In this vein, as mentioned previously, we do not know whether the expected
number of iterations for the randomized straight line procedure is inherently weakly polynomial. We 
leave this as an open problem.


\bibliographystyle{plain}
\bibliography{bib/lattices,bib/acg}


\end{document}